\documentclass{lmcs}
\pdfoutput=1

\usepackage{lastpage}
\lmcsdoi{18}{1}{34}
\lmcsheading{}{\pageref{LastPage}}{}{}%
{Dec.~01,~2020}{Feb.~25,~2022}{}

\keywords{Probabilistic Databases, Possible Worlds Semantics, Query Measurability, Relational Algebra, Aggregation, Infinite Probabilistic Databases}

\theoremstyle{plain}\newtheorem{cond}[thm]{Condition}
\theoremstyle{definition}

\usepackage[utf8]{inputenc}
\usepackage{amssymb}
\usepackage{csquotes}
\usepackage{mathtools}
\usepackage{bm}
\usepackage{float}
\usepackage{subcaption}
\usepackage{booktabs}
\usepackage{xcolor}

\usepackage[most]{tcolorbox}
\usepackage{tabularx}
\usepackage{array}
\usepackage{colortbl}
\usepackage{multirow}
\tcbuselibrary{skins}
\tcbuselibrary{breakable}
\usepackage{varwidth}
\usepackage{subcaption}
\makeatletter
\tcbset{
  tabular/.style={
    boxsep=\z@,top=\z@,bottom=\z@,leftupper=\z@,rightupper=\z@,
    toptitle=1mm,bottomtitle=1mm,boxrule=0.25mm,
    before upper={\arrayrulecolor{black}%
      \tcb@hack@currenvir\tabular{#1}},
    after upper=\endtabular\arrayrulecolor{black}},
}
\makeatother
\newcolumntype{Y}{>{\raggedleft\arraybackslash}X}
\tcbset{fancy table/.style={%
		enhanced,fonttitle=\bfseries\small,fontupper=\small,
		colback=white,colframe=black,
		colbacktitle=lgray,coltitle=black,center title,
		arc=0mm,sharpish corners
}}
\colorlet{lgray}{black!12!white}
\colorlet{llgray}{black!4!white}

\usetikzlibrary{arrows,positioning}
\tikzset{ point/.style={inner sep=0pt,circle,fill,minimum width=4pt} }

\colorlet{shade0}{black!5!white}
\colorlet{shade1}{black!10!white}
\colorlet{shade2}{black!20!white}
\colorlet{shade3}{black!50!white}

\usepackage[hidelinks]{hyperref}
\usepackage[capitalise,noabbrev]{cleveref}

\Crefname{defi}{Definition}{Definitions}
\Crefname{fact}{Fact}{Facts}
\Crefname{asm}{Assumption}{Assumptions}
\Crefname{lem}{Lemma}{Lemmas}
\Crefname{cor}{Corollary}{Corollaries}
\Crefname{prop}{Proposition}{Propositions}
\Crefname{thm}{Theorem}{Theorems}
\Crefname{itm}{Item}{Items}
\Crefname{exa}{Example}{Examples}
\Crefname{cond}{Condition}{Conditions}
\Crefname{ptys}{Properties}{Properties}
\Crefname{rem}{Remark}{Remarks}
\creflabelformat{itm}{(#2#1#3)}

\newcommand*{\centertext}[1]{\qquad\text{#1}\qquad}

\let\from\colon
\let\with\colon
\let\after\circ

\let\phi\varphi
\let\epsilon\varepsilon

\makeatletter
\foreach \x in {A,...,Z}{%
	\expandafter\xdef \csname \x \endcsname			{ \noexpand\mathcal{\x} }%
	\expandafter\xdef \csname \x\x \endcsname			{ \noexpand\mathbb{\x} }%
	\expandafter\xdef \csname \x\x\x \endcsname		{ \noexpand\mathfrak{\x} }%
	\expandafter\xdef \csname \x\x\x\x \endcsname	{ \noexpand\bm{\x} }%
}%
\makeatother

\usepackage{xparse}

\let\oldFF\FF
\RenewDocumentCommand{\FF}{o}{%
	\IfNoValueTF{#1}{\oldFF}{\oldFF\sb{[#1]}}%
}
\let\oldTT\TT
\RenewDocumentCommand{\TT}{o}{%
	\IfNoValueTF{#1}{\oldTT}{\oldTT\sb{[#1]}}%
}

\NewDocumentCommand{\DB}{o}{%
	\IfNoValueTF{#1}{\mathord{\DD\BB}}{\mathord{\DD\BB}\sb{[#1]}}%
}
\NewDocumentCommand{\DDB}{o}{%
	\mathord{\DDD\mkern-2.5mu\BBB}\IfNoValueTF{#1}{}{\sb{[#1]}}%
}

\newcommand*{\Rel}{\mathalpha{\mathsf{Rel}}}
\newcommand*{\Att}{\mathalpha{\mathsf{Att}}}

\newcommand*{\REL}[1]{\text{\normalfont\scshape #1}}
\newcommand*{\ATT}[1]{\text{\normalfont\rmfamily #1}}
\newcommand*{\STR}[1]{\text{\small\ttfamily #1}}

\newcommand*{\Borel}{\mathfrak{B\mkern-1muo\mkern-.2mur}}
\newcommand*{\Count}{\mathfrak{C\mkern-.2muo\mkern-.5muu\mkern-.4mun\mkern-.5mut}}

\DeclareMathOperator{\ar}{ar}
\DeclareMathOperator{\dom}{dom}
\DeclareMathOperator{\proj}{proj}
\DeclareMathOperator{\sort}{sort}
\DeclareMathOperator{\sym}{sym}
\renewcommand*{\restriction}{\mathord{\upharpoonright}}

\newcommand*{\counting}[2]{ \#( #1, #2 ) }

\newcommand*{\PDBs}{\bm{\mathsf{PDB}}}

\DeclarePairedDelimiter{\set}{\{}{\}}
\DeclarePairedDelimiter{\card}{\lvert}{\rvert}

\DeclarePairedDelimiterX{\dbraces}[1]{\lbrace}{\rbrace}
{%
	\nbrace{\lbrace}{#1}\delimsize\lbrace\mathopen{}%
	#1%
	\mathclose{}\delimsize\rbrace\nbrace{\rbrace}{#1}%
}
\newcommand{\dummydelim}[2]{$\left#1\vphantom{#2}\right.$}
\newcommand{\nbrace}[2]
{\sbox0{\dummydelim{#1}{#2}}\hspace{\the\dimexpr -0.85\wd0 + 2pt\relax}}
\makeatletter
	\def\bag{\@ifstar\@bag\@@bag}
	\def\@bag#1{\dbraces{\smash{#1}}}
	\def\@@bag#1{\dbraces*{#1}}
\makeatother

\newcommand*{\powerset}{\mathcal{P}}
\newcommand*{\powerbag}{\mathcal{B}}

\newcommand*{\compl}{c}
\newcommand*{\fin}{\mathsf{fin}}
\newcommand*{\quotient}[2]{#1\mkern-2mu/\mkern-1mu{#2}}

\let\epsilon\varepsilon

\DeclarePairedDelimiterXPP{\bagren}[2]{\rho_{#1}}{\lparen}{\rparen}{}{#2}

\newcommand*{\bagadd}	[2]	{\ensuremath{ #1 \uplus #2 }}
\newcommand*{\bagmax}	[2]	{\ensuremath{ #1 \cup #2 }}
\newcommand*{\bagdiff}	[2]	{\ensuremath{ #1 - #2 }}
\newcommand*{\bagmin}	[2]	{\ensuremath{ #1 \cap #2 }}

\DeclarePairedDelimiterXPP{\bagded}[1]{\delta}{\lparen}{\rparen}{}{#1}

\DeclarePairedDelimiterXPP{\bagsel}[2]{\sigma_{#1}}{\lparen}{\rparen}{}{#2}
\DeclarePairedDelimiterXPP{\bagproj}[2]{\pi_{#1}}{\lparen}{\rparen}{}{#2}
\newcommand*{\bagprod}	[2]	{\ensuremath{ #1 \times #2 }}

\DeclarePairedDelimiterXPP{\bagagg}[2]{\varpi_{#1}}{\lparen}{\rparen}{}{#2}

\newcommand*{\aggr}		[1]	{\ensuremath{ \bm{\mathsf{ #1 }} }}
\foreach\x in {CNT,CNTd,SUM,MIN,MAX,AVG}{%
	\expandafter\xdef\csname\x\endcsname{ \noexpand\aggr{\x} }%
}

\begin{document}

\title[Infinite Probabilistic Databases]{Infinite Probabilistic Databases\rsuper*}
\titlecomment{{\lsuper*}Article version of the paper \emph{Infinite Probabilistic Databases}~\cite{GroheLindner2020}
presented at ICDT 2020}

\author[M.~Grohe]{Martin Grohe}
\author[P.~Lindner]{Peter Lindner}
\address{RWTH Aachen University, Ahornstra\ss{}e 55, 52074 Aachen, Germany}
\email{\{\texttt{grohe},\texttt{lindner}\}\texttt{@informatik.rwth-aachen.de}}

\begin{abstract}
Probabilistic databases (PDBs) model uncertainty in data in a quantitative way. 
In the established formal framework, probabilistic (relational) databases are 
finite probability spaces over relational database instances. This finiteness
can clash with intuitive query behavior (Ceylan et al., KR~2016), and with 
application scenarios that are better modeled by continuous probability 
distributions (Dalvi et al., CACM~2009). 

We formally introduced infinite PDBs in (Grohe and Lindner, PODS~2019) with a
primary focus on countably infinite spaces. However, an extension beyond 
countable probability spaces raises nontrivial foundational issues concerned
with the measurability of events and queries and ultimately with the question
whether queries have a well-defined semantics.

We argue that finite point processes are an appropriate model from probability 
theory for dealing with general probabilistic databases. This allows us to
construct suitable (uncountable) probability spaces of database instances in a
systematic way. Our main technical results are measurability statements for 
relational algebra queries as well as aggregate queries and Datalog queries.
\end{abstract}

\maketitle

\section{Introduction}
Probabilistic databases (PDBs) are used to model uncertainty in data. Such
uncertainty can have various reasons like, for example, noisy sensor data, the
presence of incomplete or inconsistent information, or information gathered
from unreliable sources~\cite{Aggarwal2009,Suciu+2011}. In the standard formal
framework, probabilistic databases are finite probability spaces whose sample
spaces consist of database instances in the usual sense, referred to as
\enquote{possible worlds}.  However, this framework has various shortcomings
due to its inherent \emph{closed-world assumption}~\cite{Ceylan+2016,Ceylan+2021}, and the
restriction to finite domains. In particular, any event outside of the finite
scope of such probabilistic databases is treated as an impossible event. Yet,
statistical models of uncertain data, say, for example, for temperature 
measurements as in \cref{exa:running}, usually feature the use of continuous
probability distributions in appropriate error models. This (continuous
attribute-level uncertainty) is not expressible in the traditional PDB model.
Finite PDBs also only have limited support for tuple-level uncertainty: in
finite PDBs, all possible worlds have a fixed maximum number of tuples.
Instead, in particular with respect to an open-world assumption, we would like
to be able to model probabilistic databases without an a priori bound on the
number of tuples per instance. It is worth noting that there have been a number
of approaches to PDB systems that are supporting continuous probability
distributions, and hence going beyond finite probability spaces (see the
related works section). These models, however, lack a general (unifying) formal
basis in terms of a possible worlds semantics~\cite{Dalvi+2009}. While both
open-world PDBs and continuous probability distributions in PDBs have received
some attention in the literature, there is no systematic joint treatment of
these issues with a sound theoretical foundation. In~\cite{GroheLindner2019},
we introduced an extended model of PDBs as arbitrary (possibly infinite)
probability spaces over finite database instances. However, the focus there was
on countably infinite PDBs. An extension to continuous PDBs, which is necessary
to model probability distributions appearing in many applications that involve
real-valued measurement data, raises new fundamental questions concerning the
measurability of events and queries.

\medskip

In this paper, we lay the foundations of a systematic and sound treatment of
infinite, even uncountable, probabilistic databases. In particular, we prove
that queries expressed in standard query languages have a well-defined
semantics that is compatible with the existing theoretical point of view.
Our model is based on the mathematical theory of finite point
processes~\cite{Moyal1962,Macchi1975,DaleyVere-Jones2003}. Adopting this theory to the
context of relational databases, we give a suitable construction of measurable
spaces over which our probabilistic databases can then be defined. The only
(and mild) assumption we need is that the domains of all attributes are Polish
spaces. Intuitively, this requires them to have nice topological properties.
All typical domains one might encounter in database theory, for example 
integers, strings, and reals, satisfy this assumption. 

For queries and views to have a well-defined open-world semantics, we need them
to be measurable mappings between probabilistic databases. Our main technical
result states that indeed all queries and views that can be expressed in the
relational algebra, even equipped with arbitrary aggregate operators
(satisfying some mild measurability conditions) are measurable mappings. The
result holds for both a bag-based and set-based relational algebra and entails
the measurability of Datalog queries.

Measurability of queries may seem like an obvious minimum requirement,
but one needs to be very careful. We give an example of a simple, innocent
looking ``query'' that is not measurable (see Example~\ref{ex:nonmeas}).
The proofs of the measurability results are not trivial, and to our knowledge
not immediately covered by standard results from point process theory. At their 
core, the proofs are based on finding suitable ``countable approximations'' of 
the queries. That such approximations can be obtained is guaranteed by our
topological requirements.

In the last section of this paper, we briefly discuss queries for probabilistic
databases that go beyond those that are just set-based versions of traditional 
database queries. This also casts other natural PDB queries into our framework.
Examples of such a queries are probabilistic threshold queries and rank
queries. Note that these examples refer not only to the facts in a database,
but also to their probabilities, and hence are inherently probabilistic.

\bigskip

This article is an extended version of the paper with the same title,
\emph{Infinite Probabilistic Databases} that was presented at the 23rd 
International Conference on Database Theory (ICDT 2020)~\cite{GroheLindner2020}.
The presentation of various proofs and arguments has
been extensively reworked, and the paper contains much more background
information and new examples. The overall accessibility of the paper has been
additionally enhanced by many notational and structural improvements.  
\subsection*{Related Work}
Early models for probabilistic databases date back to the
1980s~\cite{Wong1982,GelenbeHebrail1986,CavalloPittarelli1987} and
1990s~\cite{Barbara+1992, Pittarelli1994,DeySarkar1996,FuhrRolleke1997,Zimanyi1997}.
Essentially, they are special cases or variations of the now-acclaimed formal
model of probabilistic databases as a finite set of database instances (the
\enquote{possible worlds}) together with a probability
distribution~\cite{Aggarwal2009,Suciu+2011}.

The work~\cite{Koch2008} presents a formal definition of the probabilistic
semantics of relational algebra queries which is used in the MayBMS
system~\cite{Koch2009}. A probabilistic semantics for Datalog has already been
proposed in the mid-90s~\cite{Fuhr1995}. More recently, a version of Datalog
was considered in which rules may fire probabilistically~\cite{Deutch+2010}.
Aggregate queries in probabilistic databases were first treated systematically
in~\cite{Ross+2005} and reappear in various works concerning particular PDB
systems~\cite{Murthy+2011,Fink+2012}.

The models of possible worlds semantics mentioned above are the mathematical
backbone of existing probabilistic database prototype systems such as
MayBMS~\cite{Koch2009}, Trio~\cite{Widom2009} and MystiQ~\cite{Boulos+2005}.
Various subsequent prototypes feature uncountable domains as well, such as
Orion~\cite{Singh+2008}, MCDB / SimSQL~\cite{Jampani+2011,Cai+2013},
Trio~\cite{Widom2009,AgrawalWidom2009} and PIP~\cite{KennedyKoch2010}. The MCDB
system in particular allows programmers to specify probabilistic databases with 
infinitely many possible worlds, and with database instances that can grow 
arbitrarily large~\cite{Jampani+2011}. Its description does not feature a 
general formal, measure theoretic account of its semantics though. In a spirit
similar to our work, the work~\cite{Tran+2012} introduces a measure theoretic
semantics for probabilistic data stream systems with probability measures
composed from Gaussian mixture models but (to our knowledge) on a per tuple
basis and without the possibility of inter-tuple correlations. Continuous
probabilistic databases have already been considered earlier in the context of
sensor networks~\cite{Faradjian+2002,Cheng+2003,Deshpande+2004}.  The first
work to formally introduce continuous possible worlds semantics (including
aggregation) is~\cite{Abiteboul+2011} for probabilistic XML. However, the
framework has an implicit restriction bounding the number of tuples in a PDB.

Models similar in expressivity to the one we present have also been suggested
in the context of probabilistic modeling languages and probabilistic
programming~\cite{Milch+2005,Milch2006,RichardsonDomingos2006,DeRaedt+2016,Belle2020,Gordon+2014}.
In particular notable are the measure theoretic treatments of Bayesian Logic
(BLOG)~\cite{Milch+2005} in~\cite{Wu+2018} and Markov Logic Networks (MLNs)~\cite{RichardsonDomingos2006} in~\cite{SinglaDomingos2007}. While these data
models are relational, their suitability for general database applications 
remains unclear. In particular, the investigation of typical database queries
is beyond the scope of the corresponding works. Recently, point processes in
particular have been abstractly investigated for probabilistic
programming~\cite{DashStaton2020}.

Problems caused by the closed-world assumption~\cite{Reiter1981} in
probabilistic databases~\cite{ZimanyiPirotte1997} were initially discussed by
Ceylan et al. in~\cite{Ceylan+2016} where they suggest the model of OpenPDBs.
In~\cite{Borgwardt+2018}, the authors make a more fine-grained distinction
between an \emph{open-world} and \emph{open-domain assumption}, the latter of
which does not assume the attribute values of the database schema to come from
a known finite domain. The semantics of OpenPDBs can be strengthened towards an
open-domain assumption by the means of
ontologies~\cite{Borgwardt+2017,Borgwardt+2018,Borgwardt+2019}.

There is a tight connection between the problem of computing the probability of
Boolean queries in (finite) probabilistic databases and \emph{weighted
first-order model counting (WFOMC)}~\cite{VandenBroeckSuciu2017}. For this 
connection, the \emph{lineage} of the query in question is considered. Then
probabilistic query evaluation in finite PDBs boils down to weighted counting
of the models of the lineage. While weighted model counting it itself is
restricted to discrete weight functions, \emph{weighted model integration
  (WMI)} is an extension that also supports uncountable
domains~\cite{Belle+2015,Morettin+2019}. However, a direct connection to our infinite 
PDBs seems hard to obtain, as queries in infinite PDBs in general do not have
finite lineage expressions. In particular, our infinite PDBs may have an
unbounded number of random tuples, contrasting the setup of a fixed number of
variables in a WMI problem. Also, in this paper we introduce a rather general
notion of PDB queries, that covers, for example, parity tests and various
fixpoint queries.

Our classification of views towards the end of this paper is similar to 
previous classifications of queries such
as~\cite{Cheng+2003,WandersVanKeulen2015} in the sense that it distinguishes the
level on which information is aggregated. The work~\cite{WandersVanKeulen2015}
suggests a distinction between \enquote{traditional} and \enquote{out-of-world
aggregation} similar to the one we present.

\section{Preliminaries}

Throughout this paper, $\NN$, $\QQ$, and $\RR$ denote the sets of non-negative
integers, rationals, and real numbers, respectively. With $\NN_+$, $\QQ_+$, and
$\RR_+$ we denote their restrictions to positive numbers. 

\subsection*{Sets and Bags}
If $S$ is a set and $k \in \NN$, then $\powerset_k( S )$ denotes the set of all
subsets of $S$ that have cardinality exactly $k$. The set of all \emph{finite} 
subsets of $S$ is then given as $\powerset_{ \fin }( S ) \coloneqq \bigcup_{ i
= 0 }^{ \infty } \powerset_{ k }( S )$. The set of \emph{all} subsets of $S$
(that is, the \emph{powerset} of $S$) is denoted by $\powerset( S )$.

A (finite) \emph{bag} (or \emph{multiset}) $B$ over a set $S$ is a function 
from $S$ to $\NN$, assigning a \emph{multiplicity} to every element from $S$. 
We interpret bags as collections of elements that may contain duplicates. 
For $s \in S$ we let $\card{ B }_s$ denote the multiplicity of $s$ in $B$. If
$\SSSS \subseteq S$, we let $\card{ B }_{ \SSSS } \coloneqq \sum_{ s \in \SSSS
} \card{ B }_s$. The \emph{cardinality} $\card{ B }$ of the bag $B$ is the sum
of all multiplicities, i.\,e.\ $\card{ B } \coloneqq \card{ B }_S$.  

Similar to the set notation, we use $\powerbag_{ k }( S )$ to denote the set of
all bags over $S$ that have cardinality exactly $k$. The set of all
\emph{finite} bags over $S$ is given as $\powerbag_{ \fin }( S ) \coloneqq
\bigcup_{ i = 0 }^{ \infty } \powerbag_{ k }( S )$. Occasionally, we explicitly
denote bags by the elements they contain. Then $\bag{ a_1, \dots, a_k }$
denotes a bag of cardinality $k$ with elements $a_1, \dots, a_k$ (possibly
including repetitions).

\subsection*{Relational Databases}

For the remainder of this paper, we fix two countably infinite sets $\Rel$ and
$\Att$ with $\Rel \cap \Att = \emptyset$. The elements of $\Rel$ are called
\emph{relation symbols}, and the elements of $\Att$ are called \emph{attribute 
names}.

A \emph{database schema} $\tau$ is a tuple $(\A, \R, \sort)$ such that $\A$ is 
a finite subset of $\Att$, $\R$ is a finite subset of $\Rel$, and $\sort \from
\R \to \bigcup_{ k = 0 }^{ \infty } {\A}^k$ is a function that maps every 
relation symbol $R \in \R$ to a tuple $(A_1,\dots,A_k)$ of pairwise distinct
attribute names $A_1,\dots,A_k \in \A$ for some $k \in \NN$. For $R \in \R$
we call $\sort( R )$ the \emph{sort} of $R$.  If $\sort(R) = (A_1, \dots,
A_k)$, then $k$ is called the \emph{arity} of $R$, and denoted by $\ar(R)$. We
abuse notation and write $A \in \sort( R )$ if $A$ is an attribute name
appearing in $\sort(R)$. We also write $R \in \tau$ instead of $R \in \R$ for
relation symbols.

Let $\A$ be a subset of $\Att$. A \emph{(sorted) universe} with sorts $\A$ is a
pair $(\UU, \dom)$ where $\UU$ is a non-empty set and $\dom \from \A \to
\powerset( \UU )$. We abuse notation and refer a sorted universe $( \UU, 
\dom )$ just by $\UU$, and implicitly assume $\dom$ given. For $A \in \A$, we
call $\dom( A ) \eqqcolon \AA$ the \emph{(attribute) domain} of $A$. If $\tau =
( \A, \R, \sort )$ is a database schema and $\UU$ a universe with sorts $\A$
and $R \in \R$ with $\sort( R ) = ( A_1, \dots, A_k )$, then
\begin{equation}\label{eq:domR}
	\TT[R,\UU] \coloneqq \AA_1 \times \dots \times \AA_k
\end{equation}
is called the \emph{domain} of $R$. The elements of $\TT[R,\UU]$ are called
\emph{$R$-tuples} (over $\UU$). If $R$ is a relation of sort $\sort(R) =
(A_1,\dots,A_k)$, and $t = (a_1,\dots,a_k)$ is an $R$-tuple, then 
\[
	t[ A_{i_1}, \dots, A_{i_\ell} ] \coloneqq ( a_{i_1}, \dots, a_{i_\ell} )
\]
denotes the restriction of $t$ to the attributes $A_{i_1},\dots,A_{i_\ell}$ for
all $1 \leq i_1 < \dots < i_{\ell} \leq k$ and all $\ell = 1, \dots, k$.

For $R \in \R$, an \emph{$R$-fact} (over $\UU$) is an expression of the shape
$R( t )$ where $t$ is an $R$-tuple over $\UU$. If $\TTTT \subseteq
\TT[\tau,\UU]$ is a set of $R$-tuples over $\UU$, then 
\[
	R( \TTTT ) \coloneqq \set[\big]{ R(t) \with t \in \TTTT }
	\text.
\]
In particular, $R\big( \TT[R,\UU] \big)$ is the set of all $R$-facts over
$\UU$, which is formally given as 
\begin{equation}\label{eq:domFR}
	\FF[R,\UU] \coloneqq \set{ R } \times \TT[\tau,\UU]
	\text.
\end{equation}
The elements of 
\[
	\FF[\tau,\UU] \coloneqq \bigcup_{R \in \tau} \FF[R,\UU]
\]
are called \emph{$\tau$-facts} (or just \emph{facts}) over $\UU$. We use
(variants of) the letter $f$ to denote facts.

A \emph{database instance} $D$ over $\tau$ and $\UU$, or
\emph{$\tau$-instance} over $\UU$ is a finite bag of $\tau$-facts over $\UU$. 
Thus, 
\[
	\DB[\tau,\UU] \coloneqq \powerbag_{ \fin }\big( \FF[\tau,\UU] \big)
\]
is the set of all $\tau$-instances over $\UU$. For $R \in \tau$ and $D \in
\DB[\tau,\UU]$, we let $R( D )$ denote the restriction of $D$ to $R$-facts.
That is, the instance $R(D)$ is given by
\[
	\card{ R(D) }_f \coloneqq
	\begin{cases}
		\card{ D }_f 	& \text{if } f \in \FF[R,\UU]\text{ and}\\
		0					& \text{if }f \notin \FF[R,\UU]
	\end{cases}
\]
for all $f \in \FF[\tau,\UU]$.

From all of the notation we defined above, we omit the explicit mention of
$\tau$ and $\UU$ if they are clear from the context, and just write $\TT$,
$\TT_R$, $\FF$, $\FF_R$ and $\DB$ instead of $\TT[\tau,\UU]$, $\TT[R,\UU]$,
$\FF[\tau,\UU]$, $\FF[R,\UU]$ and $\DB[\tau,\UU]$. \Cref{tab:dbnota} recaps the
notation we just defined.

\begin{table}[H]
	\caption{Notation used for relational databases.}\label{tab:dbnota}
	\begin{tabular}{ll}\toprule
		\textbf{Concept} & \textbf{Notation} \\\midrule
		Database schema 
			& $\tau = (\A, \R, \sort)$ with $\A \subseteq \Att$, $\R \subseteq \Rel$ \\
			\rowcolor{llgray}
		Attribute names
			& $A,A_1, A_2, \dots \in \A$\\
		Relation names
			& $R,S,R_1,R_2,\dots \in \R$\\
			\rowcolor{llgray}
		Sort of a relation name $R \in \R$ 
			& $\sort(R) \in \bigcup_{ k \in \NN } \A^k$\\
		Arity of a relation name $R \in \R$
			& $\ar(R) \in \NN$\\
			\rowcolor{llgray}
		Underlying universe 
			& $\UU$\\
		Space of all tuples over $(\tau, \UU)$	
			& $\TT[\tau,\UU]$, $\TT$\\
			\rowcolor{llgray}
		Space of all $R$-tuples over $(\tau,\UU)$, $R \in \R$
			& $\TT[R,\UU]$, $\TT_R$\\
		Tuples and restricted tuples 
			& $t, t_1, t_2, \dots$ and $t[A_{i_1},\dots,A_{i_k}]$\\
			\rowcolor{llgray}
		Space of all facts over $(\tau, \UU)$
			& $\FF[\tau,\UU]$, $\FF$\\
		Space of all $R$-facts over $(\tau, \UU)$, $R \in \R$
			& $\FF[R,\UU]$, $\FF_R$\\
			\rowcolor{llgray}
		Facts
			& $f,f_1,f_2,\dots$ \\
		Space of database instances over $(\tau,\UU)$
			& $\DB[\tau,\UU]$, $\DB$\\
		\bottomrule
	\end{tabular}
\end{table}

In~\cite{GroheLindner2019}, we used an example of temperature recordings to
motivate infinite probabilistic databases, to illustrate the setup, and to
highlight some observations. With some customization, we adopt this example to
also serve as a running example for this paper. 

\begin{exa}\label[exa]{exa:running}
	We consider a database that stores information about the rooms in an office
	building. An example database instance is shown in \cref{fig:temprec}. 

	\begin{figure}[H]
		\begin{tcbraster}[raster columns=2, raster force size=false, raster column skip=1cm, raster halign=center, nobeforeafter]
				\begin{tcolorbox}[hbox,fancy table, tabular={c c},title={\strut$\REL{Office}$},before upper app={\rowcolor{llgray}}]
			\ATT{RoomNo}	& \ATT{Person}		\\\hline\hline
			\STR{4108}		& \STR{Alice}		\\\hline
			\STR{4108a}		& \STR{Bob}			\\\hline
			\STR{4109}		& \STR{Charlie}	\\
		\end{tcolorbox}
		\begin{tcolorbox}[hbox,fancy table, tabular={c c c},title={\strut$\REL{TempRec}$},before upper app={\rowcolor{llgray}}]
			\ATT{RoomNo}	& \ATT{Date}					& \ATT{Temp [\textdegree{}C]}\\\hline\hline
			\STR{4108}		& \STR{2021-07-12}		& \STR{20.5}\\\hline
			\STR{4108a}		& \STR{2021-07-12}		& \STR{21.0}\\\hline
			\STR{4109}		& \STR{2021-07-12}		& \STR{21.1}\\\hline
			\STR{4109}		& \STR{2021-07-13}		& \STR{20.8}\\
		\end{tcolorbox}
		\end{tcbraster}
		\caption{A database instance $D$, storing temperature recordings in an
			office building.}\label{fig:temprec}
	\end{figure}

	The database schema $\tau$ consists of 
	\begin{itemize}
		\item relation symbols $\R = \set{ \REL{Office}, \REL{TempRec} }$, 
		\item attribute names $\A = \set{ \ATT{RoomNo}, \ATT{Person}, \ATT{Date},
			\ATT{Temp} }$, and
		\item sorts $\sort( \REL{Office} ) = ( \ATT{RoomNo}, \ATT{Person} )$ and
			$\sort( \REL{TempRec} ) = ( \ATT{RoomNo}, \ATT{Date}, \ATT{Temp} )$.
	\end{itemize}
	The sorted universe is given by $(\UU,\dom)$ with $\UU = \Sigma^* \cup \RR$
	where $\Sigma$ is some alphabet, say, the set of ASCII symbols, and (for
	simplicity) $\dom(\ATT{RoomNo}), \dom(\ATT{Person}), \dom(\ATT{Date})
	\subseteq \Sigma^*$ and $\dom( \ATT{TempRec} ) = \RR$.

	For example facts $f_1 = \REL{Office}( \STR{4108}, \STR{Bob} )$ and $f_2 =
	\REL{TempRec}( \STR{4108}, \STR{2021-07-12}, \STR{20.5} )$ are both facts
	from $\FF[\tau,\UU]$, and it holds that $f_1 \notin D$ and $f_2 \in D$ for
	the database instance $D$ from \cref{fig:temprec}.
\end{exa}

\subsection*{Measurable Spaces and Functions}

In the following we recap the relevant background and foundations from
probability and measure theory that are needed throughout this paper. For a 
general introduction to measure theory we point the reader to~\cite[Chapter 
1]{Klenke2014} or~\cite[Chapter 1]{Kallenberg2002}. 

Let $\XX \neq \emptyset$. A \emph{$\sigma$-algebra} on $\XX$ is a family $\XXX$
of subsets of $\XX$ such that $\XX\in \XXX$ and $\XXX$ is closed under
complements and countable unions (an equivalent definition is obtained by
replacing \enquote{countable unions} with \enquote{countable intersections}).

\begin{nota}
	Throughout the paper, we use double struck letters ($\XX, \YY, \ZZ, \AA, \BB,
	\dots$) to denote underlying spaces of interest, and fraktur letters ($\XXX,
	\YYY, \ZZZ, \AAA, \BBB, \dots$) to denote set families and $\sigma$-algebras
	over such spaces, in particular. We usually denote subsets of the underlying
	spaces with bold italics letters $(\XXXX, \YYYY, \ZZZZ, \AAAA, \BBBB,
	\dots)$.
\end{nota}

Let $\GGG$ (fraktur \enquote{$G$}) be a set of subsets of $\XX$ (that is, $\GGG
\subseteq \powerset( \XX )$). The \emph{$\sigma$-algebra generated by
$\GGG$} is the coarsest (i.\,e.\ smallest with respect to set inclusion)
$\sigma$-algebra $\XX$ containing $\GGG$. For any $\GGG \subseteq \powerset(
\XX )$, the $\sigma$-algebra generated by $\GGG$ is unique. A \emph{measurable space} is a pair $(\XX,\XXX)$ where $\XX$ is an arbitrary set and $\XXX$ is a $\sigma$-algebra on $\XX$. The sets in $\XXX$ are called \emph{$\XXX$-measurable} (or \emph{measurable}, if $\XXX$ is clear from the context). If $( \YY, \YYY )$ is another measurable space, then a function $\phi \from \XX \to \YY$ is called \emph{$( \XXX, \YYY )$-measurable} (or simply \emph{measurable}, if $\XXX$ and $\YYY$ are clear from the context), if for all $\YYYY \in \YYY$ it holds that 
\[
	\phi^{-1}( \YYYY ) = 
	\set[\big]{ X \in \XX \with \phi( X ) \in \YYYY } \in \XXX
	\text.
\]
The following simple properties are needed throughout the paper.

\begin{fact}[{cf.~\cite[Lemmas 1.4 and 1.7]{Kallenberg2002}}]
	\label[fact]{fac:measurablebasics}
	Let $( \XX, \XXX )$, $(\YY,\YYY)$, $(\ZZ,\ZZZ)$ be measurable spaces.
	\begin{enumerate}
		\item\label[itm]{itm:basics1}
			Suppose $\YYY$ is generated by some $\GGG \subseteq \powerset( \YY )$.
			If $\phi \from \XX \to \YY$ satisfies $\phi^{-1}( \GGGG ) \in 
			\XXX$ for all $\GGGG \in \GGG$, then $\phi$ is measurable.
		\item\label[itm]{itm:basics2}
			Compositions of measurable functions are measurable. That is, if $\phi
			\from \XX \to \YY$ and $\psi \from \YY \to \ZZ$ are measurable, then
			$\psi \after \phi \from \XX \to \ZZ$ is $(\XXX, \ZZZ)$-measurable.
		\qed
	\end{enumerate}
\end{fact}

\subsection*{Probability and Image Measures}
A \emph{probability measure} on a measurable space $(\XX, \XXX)$ is a 
countably additive function $P \from \XXX \to [0,1]$ with $P( \XX ) = 1$.
Then $( \XX, \XXX, P )$ is called a \emph{probability space}. Measurable 
sets in probability spaces are also called \emph{events}. A measurable 
function $\phi$ from a probability space $( \XX, \XXX, P )$ into a
measurable space $( \YY, \YYY )$ introduces a new probability measure $P'$
on $( \YY, \YYY )$ by
\[
	P'( \YYYY ) 
	= P\big( \phi^{-1}( \YYYY ) \big) 
	= P\big( \set{ X \in \XX \with \phi( X ) \in \YYYY } \big)
\]
for all $\YYYY \in \YYY$. Then $P' = P \after \phi^{-1}$ is called the
\emph{image} or \emph{push-forward probability measure} of $P$ under (or
along) $\phi$. The measurability of $\phi$ ensures that $P'$ is
well-defined.

\subsection*{Standard Constructions}

Throughout this paper, we will encounter a variety of measurable spaces. The
most basic of these are either well-known measurable spaces like countable
spaces with the powerset $\sigma$-algebra or uncountable spaces like $\RR$ with
its Borel $\sigma$-algebra. From these spaces, we then construct more 
complicated measurable spaces with the use of the following standard
constructions.

\begin{enumerate}
	\item \textbf{Product spaces.}
For $i = 1, \dots, n$, let $( \XX_i, \XXX_i )$ be a measurable space. The
\emph{product $\sigma$-algebra}
\[
	\bigotimes_{ i = 1 }^{ n } \XXX_i 
	= \XXX_1 \otimes \dotsc \otimes \XXX_n
\]
is the $\sigma$-algebra on $\prod_{ i = 1 }^{ n } \XX_i$ that is generated by
the sets $\set{ \proj_j^{-1}( \XXXX ) \with \XXXX \in \XXX_j }$ with $j = 1,
\dots, n$ where $\proj_j$ is the canonical projection $\proj_j \from \prod_{ i
= 1 }^{ n } \XX_i \to \XX_j$. If $( \XX_i, \XXX_i ) = ( \XX, \XXX )$ for all
$i = 1, \dots, n$, we also write $\XXX^{ \otimes n }$ instead of $\bigotimes_{
i = 1 }^{ n } \XXX$.

\item \textbf{Disjoint unions.}
If the spaces $\XX_i$ are pairwise disjoint, then the \emph{disjoint union
$\sigma$-algebra} on $\bigcup_{ i = 1 }^{ n } \XX_i$ is given by
\[
	\bigoplus_{ i = 1 }^{ n } \XXX_i \coloneqq
	\set[\big]{ 
		\XXXX \subseteq { \textstyle \bigcup_{ i = 1 }^{ n } } \XX_i \with 
		\XXXX \cap \XX_i \in \XXX_i \text{ for all } i = 1, \dots, n
	}
	\text.
\]
It may be easily verified that this is a $\sigma$-algebra on $\bigcup_{ i = 1
}^{ n } \XX_i$.

\item \textbf{Subspaces.}
Let $( \XX, \XXX )$ be a measurable space and $\XXXX \in \XXX$. Then
\[
	\XXX\rvert_{ \XXXX } 
	\coloneqq \set{ \XXXX' \cap \XXXX \with \XXXX' \in \XXX }
\]
is a $\sigma$-algebra on $\XXXX$.

\end{enumerate}

\subsection*{Polish Topological Spaces}

There are deep connections between measure theory and general topology. In 
fact, virtually everything we discuss and show in this paper relies on the
presence of certain topological properties. The central notion from topology we
need is that of \emph{Polish spaces}~\cite[Chapter~3]{Kechris1995}. In the
following, we assume familiarity with the basic topological concepts. A small 
introduction to the basic terms can be found in \cref{app:topo}. For a thorough
introduction, we refer to~\cite{Willard2004}.

A \emph{Polish space} is a topological space that is separable (i.\,e.\ there
exists a countable dense set), and completely metrizable (i.\,e.\ there is a
complete metric on the space generating its topology). Polish and (the later
introduced) standard Borel spaces are introduced and treated in detail
in~\cite[Chapter 1 \&{} 2]{Kechris1995} and~\cite{Srivastava1998}. In this paper,
we heavily exploit the properties of Polish spaces. The existence of a
countable dense set, together with the existence of a complete metric
generating the topology allows us to approximate any point in the space by
countable collections of open sets. More specifically, for every point in the
space, there exists a sequence over a fixed \emph{countable} set that converges
to the point. 

\medskip

The following fact lists a few natural classes of spaces that are
Polish.

\begin{fact} 
	All finite and countably infinite spaces (with the discrete topology) are
	Polish~\cite[p.~13]{Kechris1995}. The Euclidean space $\RR^n$ is Polish for
	all $n \in \NN_+$~\cite[ibid.]{Kechris1995}.
	
	Countable sums and products of sequences of Polish spaces are Polish
	(subject to the sum and product topologies)
	\cite[Proposition~3.3]{Kechris1995}. Moreover, subspaces of a Polish space
	that are countable intersections of open sets are Polish. This also implies
	that closed subspaces of Polish spaces are Polish, cf.~\cite[Proposition~3.7
	and Theorem~3.11]{Kechris1995}.
\end{fact}

The above properties arguably capture all typical spaces of interest for
database theory. When we later work with Polish spaces, we always assume that
we have a fixed Polish metric at hand, that is, a complete metric generating 
the Polish topology. With respect to said \emph{compatible} metric, we write
$B_{\epsilon}( X )$ to denote the open ball of radius $\epsilon > 0$ around the
point $X \in \XX$.

\subsection*{Standard Borel Spaces}\label{sec:standardborel}

If $( \XX , \OOO )$ is a topological space, the \emph{Borel $\sigma$-algebra}
$\Borel( \XX )$ on $\XX$ is the $\sigma$-algebra generated by the open sets
$\OOO$ (the topology in use is usually clear from context; in our case,
provided there is one, we always use a Polish topology on $\XX$). Sets in the
Borel $\sigma$-algebra are also called \emph{Borel sets}. In the following, we
state some nice basic properties of standard Borel spaces.

\begin{fact}[{see~\cite[Lemmas 1.5 \&{} 1.10]{Kallenberg2002}}]\label[fact]{fac:topomeas}
	Let $( \XX , \OOO_{\XX} )$ and $( \YY, \OOO_{\YY} )$ be topological spaces.
	\begin{enumerate}
		\item If $\phi \from \XX \to \YY$ is continuous, then $\phi$ is $\big(
			\Borel(\XX), \Borel(\YY) \big)$-measurable.
		\item\label[itm]{itm:topomeas2} If $(\YY, \OOO_{\YY})$ is a metric
			topological space, $(\phi_n)_{n\geq 0}$ is a sequence of $\big(
			\Borel(\XX), \Borel(\YY) \big)$-measurable functions $\phi_n \from \XX
			\to \YY$ with $\lim_{ n \to \infty } \phi_n = \phi$, then $\phi$ is
			measurable as well.
			\qed
	\end{enumerate}
\end{fact}

A measurable space $( \XX, \XXX )$ is called a \emph{standard Borel space} if
there exists a Polish topology $\OOO$ on $\XX$ such that $\XXX$ is the (Borel)
$\sigma$-algebra generated by $\OOO$.

\begin{fact}[Lusin-Souslin, see {\cite[Theorem 15.1]{Kechris1995}}]
	\label[fact]{fac:SBSinjmeasimg}
	Let $( \XX, \XXX )$ and $( \YY, \YYY )$ be standard Borel spaces, let $\XXXX 
	\in \XXX$, and let $\phi \from \XX \to \YY$ be a continuous function such
	that $\phi \restriction_{\XXXX}$ (the restriction of $\phi$ to the set $\XXXX$)
	is injective. Then $\phi( \XXXX ) \in \YYY$.\qed
\end{fact}

In other words, the above fact states that between standard Borel spaces, the
image of a measurable set $\XXXX$ under a function is measurable itself,
provided that the function is continous, and injective on $\XXXX$.

The following fact intuitively states that the property of measurable spaces
being standard Borel is closed under using the constructions introduced before.

\begin{fact}[{see~\cite[p.~75]{Kechris1995}}]
	\label[fact]{fac:SBSconstructions}
	Let $( \XX_i, \XXX_i )$ be a sequence of standard Borel spaces. 
	\begin{enumerate} 
		\item The product space $\big( \prod_{ i } \XX_i, \bigotimes_{ i } \XXX_i
			\big)$ is standard Borel.
		\item If the spaces $\XX_i$ are pairwise disjoint, then the disjoint sum
			$\big( \sum_{ i } \XX_i, \bigoplus_{ i } \XXX_i \big)$ is standard
			Borel.
	\end{enumerate}
	Moreover, if $( \XX, \XXX )$ is standard Borel, and $\XXXX \in \XXX$, then
	$( \XXXX, \XXX\rvert_{\XXXX} )$ is standard Borel.
	\qed
\end{fact}

\subsection*{(Finite) Point Processes}\label{sec:fpp}

A \emph{point process}~\cite{DaleyVere-Jones2003,DaleyVere-Jones2008} is a
probability space over countable sets of points in some abstract \enquote{state
space} such as the Euclidean space $\RR^n$. Point processes are a well-studied
subject in probability theory and they appear in a variety of applications such
as particle physics, ecology, geostatistics
(cf.~\cite{DaleyVere-Jones2003,Baddeley2007} and target tracking~\cite{Degen2017}.
In computer science, they have applications in queuing
theory~\cite{Franken1982} and machine learning~\cite{KuleszaTaskar2012}. The
individual outcomes (that is,
point configurations) drawn from a point process are called a \emph{realization
of the process}. A point process is called \emph{finite} if all of its
realizations are finite. In the following, we recall the construction of a
finite point process over a Polish state space, following the classic
constructions of~\cite{Moyal1962,Macchi1975}.  

Let $( \XX, \XXX )$ be a standard Borel space. Recall that $( \XX^k, \XXX^{ 
\otimes k } )$ is also standard Borel for all $k \in \NN_+$. For all 
$( X_1, \dots, X_k ), ( X_1', \dots, X_k' ) \in \XX^k$, we define
\[
	( X_1, \dots, X_k ) \sim_k ( X_1', \dots, X_k' )
\]
if and only if there exists a permutation $\pi$ of $\set{ 1, \dots, k }$ such
that 
\[
	( X_1, \dots, X_k ) 
	= \big( X_{ \pi(1) }', \dots, X_{ \pi(k) }' \big)\text.
\]
Then $\sim_k$ is an equivalence relation on $\XX^k$. Tuples in $\XX^k$ are
equivalent under $\sim_k$ if and only if they contain every individual element
the exact same number of times. Thus, the quotient space $\quotient{ \XX^k }{ 
\sim_k }$ can be identified with the set $\powerbag_k( \XX )$ of bags of
cardinality exactly $k$ over $\XX$. Note that in the case $k = 0$, $\quotient{
\XX^0 }{ \sim_0 }$ (by convention) consists of a single distinguished point
representing the empty realization. The space
\[
	\bigcup_{ k = 0 }^{ \infty } \quotient{ \XX^k }{ \sim_k }
	= \bigcup_{ k = 0 }^{ \infty } \powerbag_k( \XX )
	= \powerbag_{ \fin }( \XX )
\]
is one of the canonical and equivalent choices for the sample space for a 
finite point process on $\XX$~\cite{Macchi1975}. The original construction
of~\cite{Moyal1962} considers the \emph{symmetrization} $\sym \from \bigcup_{ k =
0 }^{ \infty } \XX^k \to \powerbag_{ \fin }( \XX )$ where for all $( X_1,
\dots, X_k ) \in \XX^k$ it holds that
\[
	\sym( X_1, \dots, X_k )
	= [ (X_1, \dots, X_k) ]_{ \sim_k }
	= \bag{ X_1, \dots, X_k }
	\text.
\]
Then $\bigoplus_{ k = 1 }^{ \infty } \XXX^{ \otimes k }$ is the natural
$\sigma$-algebra on $\bigcup_{ k \in \NN } \XX^k$. The $\sigma$-algebra on
$\powerbag_{ \fin }( \XX )$ is then defined as
\begin{equation}\label{eq:symsigmaalg}
	\set[\big]{
		\BBBB \subseteq \powerbag_{\fin}( \XX ) \with
		\sym^{-1}( \BBBB ) \in {\textstyle\bigoplus_{ k = 0 }^{ \infty }}
		\XXX^{ \otimes k }
	}
	\text,
\end{equation}
that is, as the $\sigma$-algebra on $\powerbag_{ \fin }( \XX )$ induced by $\sym
\from \bigcup_{ k = 0 }^{ \infty } \XX^k \to \powerbag_{ \fin }( \XX )$. Note
that this is indeed a $\sigma$-algebra on $\powerbag_{ \fin }( \XX )$,
cf.~\cite[Lemma 1.3]{Kallenberg2002}.

\begin{rem}[cf.~{\cite[Section 2]{Moyal1962}}]\label[rem]{rem:symborel}
	A set $\XXXX \in \XXX^{ \otimes k }$ is called \emph{symmetric}, if for
	all $(X_1, \dots, X_k) \in \XX^k$ it holds that 
	$(X_1, \dots, X_k) \in \XXXX$ implies $\big( X_{\pi(1)}, \dots, X_{\pi(k)}
	\big) \in \XXXX$ for all permutations $\pi$ of $\set{1, \dots, k}$. The
	function $\XXXX \mapsto \sym( \XXXX )$ is a bijection between the
	symmetric sets $\XXXX$ in $\bigoplus_{ k = 0 }^{ \infty } \XXX^{ \otimes k
	}$ and the sets of \labelcref{eq:symsigmaalg}.
\end{rem}

An equivalent, but technically more convenient construction is motivated by
interpreting point processes as
\emph{random counting measures}~\cite{DaleyVere-Jones2008}.
For $\XXXX \in \XXX$ and $n \in \NN$, let 
\[
	\counting{ \XXXX }{n} \coloneqq
	\set[\big]{ 
		B \in \powerbag_{ \fin }( \XX ) \with 
		\card{ B }_{ \XXXX } 
		= {\textstyle \sum_{ X \in \XXXX } } \card{ B }_{ X }
		= n
	}
	\text,
\]
that is, $\counting{ \XXXX }{ n }$ is the set of finite bags over $\XX$
containing exactly $n$ elements from $\XXXX$, counting multiplicities. The sets 
$\counting{ \XXXX }{ n }$ with $\XXXX \in \XXX$ and $n \in \NN$ are called
\emph{counting events}. The \emph{counting $\sigma$-algebra} $\Count( \XX )$ 
on $\XX$ is the $\sigma$-algebra generated by all counting events.

\begin{fact}[{\cite[Theorem 3.2]{Moyal1962}}]\label[fact]{fac:symcounting}
	The $\sigma$-algebra from \labelcref{eq:symsigmaalg} coincides with 
	$\Count( \XX )$.
	\qed
\end{fact}

\begin{defi}\label[defi]{def:fpp}
	A \emph{finite point process} with standard Borel state space $( \XX, \XXX 
	)$ is a probability space $\big( \powerbag_{\fin}( \XX ), \Count( \XX ), P 
	\big)$. A finite point process is \emph{simple} if $P\big( \powerset_{ 
	\fin }( \XX )\big) = 1$, that is, if its realizations are almost surely
	sets.
\end{defi}

\section{Probabilistic Databases}

In this section, we introduce our framework for infinite probabilistic
databases and their query semantics. To begin with, we recall the conventional 
formal definition of probabilistic databases as it is found in textbooks on
the subject~\cite{Suciu+2011,VandenBroeckSuciu2017}:

\begin{defi}[Finite Probabilistic Databases; adapted from {\cite[Section 2.2]{Suciu+2011}}]\label[defi]{def:finpdb}
	Let $\tau$ be a database schema and let $\UU$ be a universe. A 
	\emph{probabilistic database (PDB)} $\D$ over $\tau$ and $\UU$ is a 
	probability space $\D = \big(\DB, \powerset( \DB ), P\big)$ where $\DB$ is a
	finite set of database instances over $\tau$ and $\UU$.
\end{defi}

From a probability theoretic point of view, \cref{def:finpdb} is a severe 
limitation to the vast expressive power of stochastic models, only allowing
probability distributions over finitely many alternative database instances. As
an example, in a setting such as that of \cref{exa:running} (the database of
temperature measurements), we would typically model noise or uncertainty in the
sensor measurements by continuous distributions. This directly leads to (even
uncountably) infinite probability spaces that are not covered by the typical
textbook definition of probabilistic databases (\cref{def:finpdb}).
In~\cite{GroheLindner2019}, we introduced the following general notion of
probabilistic databases as probability spaces of database instances:

\begin{defi}[{see~\cite[Definition 3.1]{GroheLindner2019}}]\label[defi]{def:pdb}
	Let $\tau$ be database schema and let $\UU$ be a universe. Suppose that 
	$\FFF$ is a $\sigma$-algebra on the space $\FF[\tau,\UU]$ of all
	$(\tau, \UU)$-facts. A \emph{probabilistic database (PDB)} $\D$ over $\tau$
	and $\UU$ is a probability space $\D = ( \DB, \DDB, P )$ where $\DB
	\subseteq \DB[\tau, \UU]$ and $\DDB$ is a $\sigma$-algebra on $\DB$
	satisfying
	\[
		\set[\big]{ D \in \DB \with \card{ D }_{\FFFF } > 0 } \in \DDB
	\] 
	for all $\FFFF \in \FFF$.
\end{defi}

Therein, the sample space $\DB$ may be infinite, even uncountable. While we 
only discussed set PDBs in~\cite{GroheLindner2019}, we broaden the definition
to support bag instances here. In any case, it was left open
in~\cite{GroheLindner2019}, how to construct such probability spaces, let alone 
how to obtain suitable measurable spaces $( \DB, \DDB )$. This is
no longer a trivial task once $\DB$ is uncountable. In this section, we 
provide a general construction for such measurable spaces.

\begin{rem}
	At this point, we want to stress again the meaning of our terminology. In
	general probabilistic databases there are lots of components, on different
	levels of abstraction that could, in principle, be infinite spaces. The term
	\emph{infinite} probabilistic databases is derived from the notion of an
	\emph{infinite} probability space, meaning that the sample space is of
	infinite size. Still, in the PDBs we consider in this paper, database
	instances themselves (the concrete outcomes, or realizations of a PDB) are
	always finite collections of facts. The framework we discribe here is not
	suitable for discussing probability spaces over infinite database instances 
	in the sense of~\cite[Section 5.6]{Abiteboul+1995}, such as \emph{constraint 
	databases}~\cite{Kuper+2000}.
\end{rem}

\subsection{Probabilistic Databases are Finite Point
Processes}\label{ssec:pdbfpp}

Our counstruction of measurable spaces of database instances stems from
interpreting a probabilistic database as a finite point process over the space
of possible facts.

\begin{conv}
	\emph{From now on, we only consider database schemas $\tau$ in combination 
	with sorted universes $\UU$ where all attribute domains are Polish.}
\end{conv}

We first construct the measurable space of facts over $\tau$ and $\UU$. For
all $A \in \A$, by assumption, the domain $\AA \coloneqq \dom( A )$ is Polish.
Equipping it with its Borel $\sigma$-algebra $\AAA$ thus yields a standard 
Borel space $( \AA, \AAA )$. Now let $R \in \tau$ be a relation symbol with
$\sort( R ) = ( A_1, \dots, A_k )$, so that the standard Borel spaces belonging
to the attribute names $A_1, \dots, A_k$ are $(\AA_1,\AAA_1), \dots, (\AA_k,
\AAA_k)$. Recall from \labelcref{eq:domR}, that the set $\TT_R$ of $R$-tuples
is the product of the attribute domains in $R$. Naturally, $\TT_R$ is equipped
with its product $\sigma$-algebra 
\begin{equation}\label{eq:TTT_R}
	\TTT_R \coloneqq \bigotimes_{ i = 1 }^{ k } \AAA_i\text.
\end{equation}
Likewise, the set $\FF_R$ of $R$-facts is equipped with the
$\sigma$-algebra 
\[
	\FFF_R \coloneqq \set[\big]{ \emptyset,\set{R} } \otimes \TTT_R
\]
(cf.~\labelcref{eq:domFR}). As the set $\FF$ of all facts (over $\tau$ and
$\UU$) is the disjoint union of all the $\FF_R$, it is naturally equipped with
the disjoint union $\sigma$-algebra 
\[
	\FFF \coloneqq \bigoplus_{ R \in \tau } \FFF_R\text.
\]

\begin{exa}\label[exa]{exa:running2}
	We reenact these definitions for our running example (\cref{exa:running}).
	Recall that the attribute names $A \in \set{ \ATT{RoomNo}, \ATT{Person},
	\ATT{Date} }$, have domain $\Sigma^*$ for some (finite) alphabet $\Sigma$.
	As $\Sigma^*$ is countably infinite, we equip this space with its powerset
	$\sigma$-algebra, $\powerset( \Sigma^* )$. For the attribute 
	$\ATT{TempRec}$, we assume the measurable space to be $( \RR, \Borel(\RR)
	)$, where $\Borel(\RR)$ is the Borel $\sigma$-algebra on $\RR$. Given the
	sorts of the relation names $\REL{Office}$ and $\REL{TempRec}$, we have the
	tuple spaces $\TT_{\REL{Office}} = \Sigma^* \times \Sigma^*$ and 
	$\TT_{\REL{TempRec}} = \Sigma^* \times \Sigma^* \times \RR$. Following 
	\eqref{eq:TTT_R}, the corresponding $\sigma$-algebras are
	\[
		\TTT_{\REL{Office}}
		= \powerset( \Sigma^* ) \otimes \powerset( \Sigma^* )
	\]
	and
	\[
		\TTT_{\REL{TempRec}}
		= \powerset( \Sigma^* ) \otimes \powerset( \Sigma^* ) \otimes \Borel(
		\RR )\text.
	\]
	That is, $\TTT_{\REL{Office}}$ is the $\sigma$-algebra generated by all 
	events of the shape $L_1 \times L_2$ where $L_1,L_2 \subseteq \Sigma^*$, and
	$\TTT_{\REL{TempRec}}$ is the $\sigma$-algebra generated by all events of
	the shape $L_1 \times L_2 \times B$ where $L_1,L_2 \subseteq \Sigma^*$ and
	$B$ is a Borel set in $\RR$. The spaces of $\REL{Office}$- and
	$\REL{TempRec}$-facts are given by 
	\[
		\FF_{\REL{Office}} =
		\REL{Office}\big( \TT_{\REL{Office}} \big) = 
		\set[\big]{ \REL{Office}( r, s ) \with r,s \in \Sigma^* }
	\]
	and
	\[
		\FF_{\REL{TempRec}} =
		\REL{TempRec}\big( \TT_{\REL{TempRec}} \big) =
		\set[\big]{ \REL{TempRec}( r, d, \theta ) \with r,d \in \Sigma^*
		\text{ and } \theta \in \RR }
		\text.
	\]
	For example, (assuming that the alphabet $\Sigma$ contains the respective
	symbols), it holds that $\REL{Office}(\STR{4108},\STR{Bob}) \in
	\FF_{\REL{Office}}$, and
	$\REL{TempRec}(\STR{4108},\STR{2021-07-12},\STR{20.5}) \in
	\FF_{\REL{TempRec}}$. The $\sigma$-algebras on these spaces of facts are
	directly obtained from the $\sigma$-algebras on the tuple spaces: 
	\(
		\FFF_{\REL{Office}} 
		= \powerset\big( \set{ \REL{Office} } \big) \otimes \TTT_{\REL{Office}}
	\) 
	and 
	\(
		\FFF_{\REL{TempRec}}
		= \powerset\big( \set{ \REL{TempRec} } \big) \otimes \TTT_{\REL{TempRec}}
	\).
	Finally, 
	\[
		\FF = \FF_{\REL{Office}} \cup \FF_{\REL{TempRec}}
	\]
	is the space of all facts, with $\sigma$-algebra 
	\[
		\FFF = \FFF_{\REL{Office}} \oplus \FFF_{\REL{TempRec}}\text.
	\]
	For example, consider the sets
	\begin{align*}
		\FFFF_1 &= 
		\set[\big]{ f \in \FF_{\REL{Office}} \with f = 
		\REL{Office}( \STR{4108}, s ) \text{ for some }s \in \Sigma^* }
		\shortintertext{and}
		\FFFF_2 &=
		\set[\big]{ f \in \FF_{\REL{TempRec}} \with f = 
			\REL{TempRec}( \STR{4108}, d, \theta ) \text{ for some } d \in
			\Sigma^* \text{ and } \theta \in [21,23] }
		\text.
	\end{align*}
	Then $\FFFF_1 \in \FFF_{\REL{Office}}$ and $\FFFF_2 \in
	\FFF_{\REL{TempRec}}$. Note that due to the real interval, $\FFFF_2$
	contains uncountably many facts. As $\FFFF_1 \in \FFF_{\REL{Office}}$ and
	$\FFFF_2 \in \FFF_{\REL{TempRec}}$, it holds that $\FFFF_1 \cup \FFFF_2 \in
	\FFF$. 
\end{exa}

The constructions of the various measurable spaces above all started from
standard Borel measurable spaces for the attribute domains, and then used the
basic constructions of product and disjoint union measurable spaces. Therefore,
\cref{fac:SBSconstructions} immediately yields the following statement.

\begin{lem}
	The spaces $\big( \TT_R, \TTT_R \big)_{ R \in \tau }$, $\big( \FF_R,
	\FFF_R \big)_{ R \in \tau }$, and $( \FF, \FFF )$ are standard Borel spaces.
\end{lem}

Recall that $\DB = \DB[\tau,\UU] = \powerbag_{ \fin }( \FF )$. We equip 
$\DB$ with the $\sigma$-algebra $\DDB \coloneqq \Count( \FF )$ from
\cref{sec:fpp}.

\begin{defi}
	A \emph{standard probabilistic database (standard PDB)} over $\tau$ and
	$\UU$ is a probability space $( \DB, \DDB, P )$ with $\DB = \powerbag_{ \fin
	} ( \FF )$ and $\DDB = \Count( \FF )$.
\end{defi}

That is, a standard PDB is a finite point process over the state space $( \FF,
\FFF )$. Note that every standard PDB is a PDB in the sense of \cref{def:pdb}.
Standard PDBs that are simple are suitable for modeling PDBs with set 
semantics.

\begin{exa}
	We continue from \cref{exa:running2}. The space $\DB$ of database instances
	is exactly the space of finite bags over $\FF = \FF_{\REL{Office}} \cup
	\FF_{\REL{TempRec}}$ with the counting $\sigma$-algebra $\DDB =
	\Count(\FF)$. This means, for example that the set
	\[
		\counting{ \FFFF_2 }{ 3 } 
		= \set[\big]{ D \in \DB \with \card{ D }_{ \FFFF_2 } = 3 }
	\]
	is measurable (i.\,e., an event) in $(\DB,\DDB)$. This is the set of
	database instances over $\FF$ that contain exactly three facts of the shape
	$\REL{TempRec}( \STR{4108}, d, \theta )$ where $d$ is an arbitrary string
	and $21 \leq \theta \leq 23$. Note that these three facts need not be 
	distinct, as facts are allowed to be present with multiplicities. A standard
	probabilistic database of the example schema is just a probability space
	with underlying measurable space $(\DB,\DDB)$. In particular, in these
	PDBs, events such as $\counting{ \FFFF_2 }{ 3 }$ carry a probability.
\end{exa}

\begin{lem}
	A standard PDB $( \DB, \DDB, P )$ is a standard Borel space itself.
\end{lem}

\begin{proof}
	In~\cite[Proposition 9.1.IV]{DaleyVere-Jones2008}, it is shown that the 
	space of $\NN \cup \set{ \infty }$-valued counting measures on a standard
	Borel space (with the property of being finite on bounded sets) is a
	Polish space, whose $\sigma$-algebra is generated by the functions that map
	a counting measure $\mu$ to $\mu( \FFFF ) \in \NN \cup \set{ \infty }$ for
	all $\FFFF \in \FFF$. Restricting this space to integer-valued counting 
	measures, and equipping it with the corresponding subspace $\sigma$-algebra
	yields a standard Borel space. This space is isomorphic to $(\DB, \DDB)$
	via the function that maps a counting measure $\mu$ to the database instance
	whose multiplicity mapping is given by $\mu$.\footnote{Isomorphic here means
	that this function is bijective and measurable both ways; bijectivity is
	clear, and being measurable both ways stems from the fact that the
	generating events of the $\sigma$-algebra of the space of integer-valued
	counting measures on $(\FF, \FFF)$ are identified with the counting events in 
	$(\DB,\DDB)$, i.\,e.\ the generating events of $\DDB$.}
\end{proof}

\begin{conv}
	\emph{From now on, all PDBs we consider are standard PDBs.} When we speak of 
	PDBs, it is understood to refer to standard PDBs, unless explicitly stated
	otherwise.
\end{conv}

\subsection{The Possible Worlds Semantics of Queries and Views}\label{ssec:pws}

\emph{Views} are mappings between database instances. That is, a view $V$ is a
function $V \from \DB[\tau,\UU] \to \DB[\tau',\UU']$ for some database
schemas $\tau$ and $\tau'$ and universes $\UU$ and $\UU'$. We call $\tau$ the
\emph{input}, and $\tau'$ the \emph{output schema} of the view $V$. If $\tau'$
consists of a single relation symbol only, we call $V$ a \emph{query}. Queries 
are typically denoted by $Q$. Usually, queries and views are given as syntactic
expressions in some query language. As usual, we blur the distinction between a
query or view, and its syntactic representation. In the following, we let $(
\DB, \DDB ) = \big( \DB[\tau,\UU], \DDB[\tau,\UU]\big)$, and $( \DB_V,
\DDB_V ) = \big( \DB[\tau',\UU'], \DDB[\tau',\UU'] \big)$.

Let $\D = ( \DB, \DDB, P )$ be a PDB and let $V \from \DB \to \DB_V$ be a view.
For $\DDDD \subseteq \DB$, the image of $\DDDD$ under $V$ is
\[
	V( \DDDD )
	\coloneqq
	\set{ V( D ) \colon D \in \DDDD }
	\subseteq
	\DB_V\text.
\]
If the function $V$ is $( \DDB, \DDB_V )$-measurable, then $V$ introduces a
push-forward measure $P_V$ on $( \DB_V, \DDB_V )$ via
\begin{equation}\label{eq:pdb-push-forward}
	P_V( \DDDD ) 
	\coloneqq
	P\big( V^{-1}( \DDDD ) \big)
	=
	P\big( \set{ D \in \DB \with V( D ) \in \DDDD } \big)
\end{equation}
for all $\DDDD \in \DDB_V$. In this case
\[
	V( \D ) \coloneqq \big( \DB_V, \DDB_V, P_V \big)
\]
is a PDB. 

\begin{rem}\label[rem]{rem:pws}
	The kind of semantics we introduce here is the natural generalization of a
	standard choice for semantics on probabilistic databases. Conceptually, we
	consider queries and views that have well-defined semantics on traditional
	database instances. That is, they get as input a database instance, and as
	output, produce a new database instance. Such a semantics is lifted to
	probabilistic databases by applying the query or view on every possible
	world, and weighting it according to the probability measure of the input
	PDB. This notion of semantics is commonly called the \emph{possible answer
	sets semantics} of probabilistic databases~\cite[Section 2.3.1]{Suciu+2011}.
	It (or, to be more precise, its discrete version) has previously been also
	called \emph{possible worlds semantics (of queries)}~\cite{DalviSuciu2007},
	which is the term we prefer to use in this work, as we deem it the natural
	choice on how to define query (or view) semantics on PDBs that are modelled
	as a collection of possible worlds with a probability distribution, matching
	the standard definition of output probabilities of a (measurable) function
	on a probability space (cf.~\cite{Green2009}). Note that strictly speaking,
	this overloads the term \enquote{possible worlds semantics}: in reference to
	PDBs, \enquote{possible worlds semantics} means the definition of PDBs as 
	probability spaces over database instances, whereas in reference to 
	queries or views, it means the definition of the output of a query with
	respect to an application per possible world, as in
	\eqref{eq:pdb-push-forward}.
	
	For queries, another semantics has been discussed in literature, wich was
	later dubbed the \emph{possible answers semantics}~\cite[Section
	2.3.2]{Suciu+2011}. Under this semantics, the output of a query is the
	collection of tuples that may appear as an answer to the query (i.\,e.\ the
	tuples that appear in the output possible worlds under the possible worlds
	semantics), together with their marginal probability. For finite PDBs, this
	notion makes sense, because the result will be much smaller than a
	description of the whole output probability under possible worlds semantics.
	For uncountable infinite PDBs, however, this is not of much use. As soon as
	continuous probability distributions are involved, we naturally encounter
	PDBs where the marginal probability of every particular fact (or tuple in
	the output) may be zero.

	We note that to the best of our knowledge, there has been no formal 
	description of these semantics when duplicates are allowed. Note that for
	\emph{Boolean queries} with set semantics, i.\,e., queries whose output is
	either $\set{()}$ (true) or $\emptyset$ (false), both of the above semantics
	are essentially equivalent: the only possible answer tuple is the empty
	tuple $()$, and the only possible worlds of the answer are $\emptyset$ and
	$\set{ () }$.
\end{rem}

Note that if $V$ fails to be measurable, then
\labelcref{eq:pdb-push-forward} is not well-defined. In this case, $V$ has no
meaningful semantics on probabilistic databases! Thus, discussing the 
measurability of views and queries is an issue that requires attention. The 
following example shows that there are inconspicuous, seemingly simple queries
that are not measurable.

\begin{exa}\label[exa]{ex:nonmeas}
	Consider $\UU = \UU' = \RR$, together with the database schemas $\tau =
	\tau'$ consisting of the single, unary relation symbol $R$ with domain
	$\TT_R = \RR$ (equipped with the Borel $\sigma$-algebra). Let $\BBBB \in
	\Borel( \RR^2 )$. We define a function $Q_{ \BBBB } \from \DB \to \DB'$
	(where $\DB' = \DB[\tau',\UU']$) as follows.
	\[
		Q_{ \BBBB }( D ) \coloneqq
		\begin{cases}
			D & \text{if } D \text{ is a singleton } \bag{ R(x) } 
			\text{ and there exists a } y \in \RR \text{ s.\,t. } 
			(x,y) \in \BBBB\text{ and}\\
			\emptyset & \text{otherwise.}
		\end{cases}
	\]

	Observe that $Q_B^{-1}\big( \DB' \big) = \set[\big]{ \bag{ R(x) } \with x 
	\in \proj_1( \BBBB ) }$, where
	\[
		\proj_1( \BBBB ) = \set{ x \in \RR \with \text{ there is } y \in \RR 
		\text{ s.\,t. } (x,y) \in B }
		\text.
	\] 
	It is well known that there are Borel sets $\BBBB \subseteq \RR^2$ with the
	property that $\proj_1( \BBBB )$ is \emph{not} Borel in $\RR$~\cite[Theorem
	14.2 (Souslin) and Exercise 14.3]{Kechris1995}. In this case, $Q_{ \BBBB }$
	is not measurable even though $\BBBB$ is Borel.
\end{exa}

\subsection{Assembling Views from Queries}

A view $V$ with output relations $R_1, \dots, R_k$ can be identified with a set
of queries $Q_1, \dots, Q_k$, one per output relation symbol, such that for all
$D \in \DB$ it holds that
\[
	V( D ) = \bigcup_{ i = 1 }^{ k } Q_i( D )
	\text.
\]
Let  $(\FF_V, \FFF_V)$ be the measurable space of facts belonging to $(\DB_V,
\DDB_V)$ and let $(\FF_{Q_i}, \FFF_{Q_i})$ be the measurable space of facts
belonging to $(\DB_{Q_i}, \DDB_{Q_i})$. From the above equality, we get that
for all $D \in \DB$, all $\FFFF \in \FFF_V$, and all $n \in \NN$ it holds that
$\card{ V(D) }_{ \FFFF } = n$ if and only if there exist non-negative integers
$n_1,\dots,n_k$ with $n_1 + \dots + n_k = n$ and the property that
$\card{Q_i(D)}_{ \FFFF \cap \FF_{Q_i} } = n_i$ for all $i = 1, \dots, k$.
Note that this condition corresponds to an event given by a countable union of
counting events. Thus, we obtain the following.

\begin{lem}\label[lem]{lem:assembled-view}
	The view\/ $V$ is measurable if and only if\/ $Q_i$ is measurable for all $i
	= 1, \dots, k$.
\end{lem}

By the merit of \cref{lem:assembled-view}, we only need to discuss the 
measurability of queries.

\section{General Measurability Criteria}

In the remainder of the paper we establish measurability results for various
types of queries as they typically appear in database applications. In this
section, we set out the general setup of said investigation and introduce some
general measurability results that are not yet tailored to specific query 
languages.

\subsection{Setup}

Henceforth, we adhere to the following notational conventions when discussing
the measurability of a query $Q$. 

\begin{conv}[Inputs]
	The input schema of $Q$ is $\tau = ( \A, \R, \sort )$. We consider input 
	instances over $\tau$, and the sorted universe $\UU$ (with all attribute 
	domains Polish). 

	The associated fact space is denoted as $( \FF, \FFF )$, with subspaces 
	$( \FF_R, \FFF_R )$ for all $R \in \tau$. The space of $R$-tuples is given
	as $( \TT_R, \TTT_R )$. For all $R \in \tau$, we fix a compatible Polish 
	metric $d_R$ on $\TT_R$ and let $\TT_R^*$ be a countable, dense set in
	$\TT_R$. With abuse of notation, we denote the corresponding metric on
	$\FF_R$ by $d_R$ as well. Note that $\FF^*_R \coloneqq \TT_R^*$ is a
	countable dense set in $\FF_R$.

	We denote the input (standard) PDB under consideration by $\D = ( \DB, \DDB,
	P )$, where $\DB = \powerbag_{\fin}( \FF )$ and $\DDB = \Count( \FF )$.
\end{conv}

\begin{conv}[Outputs]
	The output schema of $Q$ is $\tau_Q = ( \A_Q, \R_Q, \sort_Q )$ where $R_Q$
	is the only relation symbol in $\R_Q$. We consider output instances over
	$\tau_Q$, and sorted universe $\UU_Q$ (with all attribute domains Polish). 
	The associated fact space is denoted as $( \FF_Q, \FFF_Q )$. The space of
	$R_Q$-tuples is given as $( \TT_Q, \TTT_Q )$. We fix a compatible Polish 
	metric $d_Q$ on $\TT_Q$, and a countable dense set $\TT_Q^*$ in $\TT_Q$.
	Again $d_Q$ will also denote the corresponding metric on $\FF_Q$, and the
	set $\FF_Q^*$ is countable and dense in $\FF_Q$.

	The output measurable space is denoted by $( \DB_Q, \DDB_Q )$, were $\DB_Q =
	\powerbag_{\fin}( \FF_Q )$ and $\DDB_Q = \Count( \FF_Q )$.
\end{conv}

Thus, our goal is to show that a given function $Q \from \DB \to \DB_Q$ is
$(\DDB, \DDB_Q)$-measurable.

\begin{rem}
	We have some straightforward measurability criteria from the general 
	properties of measurable functions and the used $\sigma$-algebras.
	\begin{enumerate}
		\item By \cref{fac:measurablebasics}\labelcref{itm:basics1}, to show the
			measurability of a query $Q$, it suffices to show that 
			\[
				\set{ D \in \DB \with \card{ Q(D) }_{ \FFFF } = n } \in \DDB
			\]
			for all $\FFFF \in \FFF_Q$ and all $n \in \NN$, as the counting
			events generate the $\sigma$-algebra $\DDB_Q$. This remains true, if
			we replace \enquote{$=n$} with \enquote{$\geq n$} or only require $n
			\in \NN_+$.
		\item By \cref{fac:measurablebasics}\labelcref{itm:basics2}, compositions
			of measurable functions are measurable. That is, for query languages
			whose semantics are defined inductively over the structure of their
			syntactic expression, it suffices to show measurability for the basic
			building blocks.
		\item By \cref{fac:topomeas}\labelcref{itm:topomeas2}, limits of 
			measurable queries are measurable.
	\end{enumerate}
\end{rem}

\subsection{The Mapping Theorem}\label{sec:mapping}
The following is a partial restatement of what is known as the \emph{mapping 
theorem} of point processes. The original theorem from point process theory 
also involves the transfer of certain properties to the image
space~\cite[Theorem 5.1]{LastPenrose2017} which is, however, of less importance for
the remainder of the paper. Moreover, we allow partial transformations as long
as their domain is measurable (and this does not invalidate the measurability
statement from the mapping theorem).

\begin{thm}[{cf.~\cite[Theorem 5.1]{LastPenrose2017}}]\label[thm]{thm:mapping}
	Let $\FFFF_q \in \FFF$. If $q \from \FFFF_q \to \FF_Q$ is
	$(\FFF,\FFF_Q)$-measurable, then the function $Q \from \DB \to \DB_Q$ with
	\[
		\card{ Q( D ) }_{ f } \coloneqq
		\begin{cases}
			\card{ D }_{ q^{-1}( f ) }	
				& \text{if }q^{-1}( f ) \neq \emptyset\text{ and}\\
			0									
				& \text{otherwise,}
		\end{cases}
	\]
	for all $D \in \DB$ is a measurable query.
	\qed
\end{thm}

In this (restricted) form, the theorem is straightforward to verify.

\begin{proof}
	Let $\FFFF_q \in \FFF$ and $q \from \FFFF_q \to \FF_Q$ be measurable.
	Now let $\FFFF \in \FFF_Q$ and $n \in \NN$. Then
	\[
		\card{ Q(D) }_{ \FFFF } = n \iff \card{ D }_{ q^{-1}( \FFFF ) } = n\text.
	\]
	Since $q^{-1}( \FFFF ) \in \FFF$, the claim follows.
\end{proof}

Intuitively, the theorem states that whenever we have a measurable
transformation of the fact space of a a PDB, then we obtain a measurable query
when we just apply this transformation to all facts in the database
instances.\footnote{Functions of the shape of $Q$ in \cref{thm:mapping} are a
special case of \enquote{mapping constructs} (applying a function to every
element of a bag) that can be found in previously considered bag query
languages~\cite{GrumbachMilo1996,LibkinWong1997}.}

\begin{exa}\label[exa]{exa:running3}
	We continue our running example. Recall that $\FF_{\REL{TempRec}}$ is the
	space of facts $\REL{TempRec}( r,d,\theta )$ where $r$ and $d$ are strings
	and $\theta$ is a real number. Consider the function $q \from
	\FF_{\REL{TempRec}} \to \FF_{\REL{TempRec}}$ that increases the temperature
	by $2$ \textdegree{}C, i.\,e.\ with
	\[
		q\big( \REL{TempRec}( r,d, \theta ) \big) =
		\REL{TempRec}( r,d, \theta + 2 )
		\text.
	\]
	Then $q$ is $( \FFF_{\REL{TempRec}}, \FFF_{\REL{TempRec}} )$-measurable. 
	This follows, since the addition of $2$ is a continuous function on $\RR$ 
	and by the construction of the measurable spaces.
	
	Thus, by \cref{thm:mapping}, the query that, given an instance $D$, applies
	$q$ to every fact in $D$ (i.\,e.\ increasing temperatures by $2$) is
	measurable.
\end{exa}

While \cref{thm:mapping} is a nice statement, it fails to cover most queries of
interest, as database queries often consider and manipulate multiple tuples at 
once. Such transformations are not captured by \cref{thm:mapping}. We therefore
need measurability statements beyond \cref{thm:mapping}.

\subsection{Continuous One-to-One Decompositions}
In this subsection we introduce a new criterion for query measurability that 
overcomes the aforementioned limitation of \cref{thm:mapping}.

\begin{lem}\label[lem]{lem:central}
	Let $Q \from \DB \to \DB_Q$. Then $Q$ is measurable if there exists some 
	$k \in \NN_{+}$, pairwise distinct $R_1, \dots, R_k \in \tau$ and functions
	$q_i \from \FF_Q \to \FF_{R_i}$ for $i = 1, \dots, k$ with the following
	properties:
	\begin{enumerate}
		\item \label[itm]{itm:central1} 
			For all $n \in \NN_{+}$ there is a set $N_Q( n ) \in \NN^k \setminus 
			\set{ (0,\dots,0) }$ with
			\[
				\card{ Q(D) }_{ f } = n
				\iff
				\big(
					\card{ D }_{ q_1(f) },
					\dotsc,
					\card{ D }_{ q_k(f) }
				\big) \in N_Q( n )
			\]
			for all $D \in \DB$ and all $f \in \FF_Q$.
		\item \label[itm]{itm:central2}
			For all $i = 1, \dots, k$, the function $q_i$ is injective and
			continuous.
	\end{enumerate}
\end{lem}

Think of the functions $q_i$ as providing a decomposition of $R_Q$-facts into
the $R_i$-facts of the input instance they originated from under the query.
The set $N_Q( n )$ provides the recipe, how the number of occurrences of a fact in
the output is determined by the counts of its decompositions in the input. That
is, requirement \labelcref{itm:central1} stipulates the query semantics. The
topological requirement \labelcref{itm:central2} ensures measurability. 

A simple example application that we ask the reader to have in mind is that
of a difference operator on bag instances (this is in fact, one of our later
applications, where it is made precise). Intuitively, if $Q$ is the difference
of relations $R$ and $S$, then the number of times a tuple $t$ occurs in the
output is given by the number of times it appears in $R$, minus the number of
times it appears in $S$. This fits the pattern of \cref{lem:central} with
functions $q_1,q_2$ being the identity $t \mapsto t$, and $N_Q(n)$ being the
set of pairs $(n_1,n_2)$ with $\max(0,n_1 - n_2) = n$. \Cref{lem:central}
provides a generalization of this setup, allowing for much more general
functions $q_i$, and sets $N_Q(n)$.

\begin{rem}
	This is vaguely related to the notion of \emph{(why- and how-)provenance}
	of a tuple in the output of a view~\cite{Green+2007}, with the functions
	$q_i$ providing the \enquote{why-information}, and the set $N_Q( n )$ 
	providing the \enquote{how}. \Cref{lem:central} now only applies to
	queries inducing a very particular provenance structure (as governed by
	the $q_i$ and $N_Q(n)$).
\end{rem}

\begin{proof}[Proof of \cref{lem:central}]
	Let $Q \from \DB \to \DB_Q$, let $k \in \NN_{+}$. Let $R_1, \dots, R_k$ be
	distinct relation symbols in $\tau$. In the following, we write $( \FF_i,
	\FFF_i )$ instead of $( \FF_{ R_i }, \FFF_{ R_i } )$ for all $i = 1, \dots,
	k$. As required, let $q_i \from \FF_Q \to \FF_i$ for all $i = 1, \dots, k$
	and assume that \labelcref{itm:central1,itm:central2} hold. We have to
	show that $Q$ is measurable, which is settled, in particular, by showing that
	$\set{ D \in \DB \with \card{ Q(D) }_{ \FFFF } \geq n } \in \DDB$ for all 
	$\FFFF \in \FFF_Q$ and all $n \in \NN_+$.

	Let $d$ be a fixed Polish metric on $\FF_Q$ generating $\FFF_Q$ and let
	$\FF_Q^* \subseteq \FF_Q$ be countable and dense in $\FF_Q$. Let $\FFFF \in
	\FFF_Q$ and $n \in \NN_{+}$. We show that for all $D \in \DB$, it holds that
	$\card{ Q(D) }_{ \FFFF } \geq n$ is equivalent to the following condition.

	\begin{cond}\label[cond]{cond:central}
		For some $\ell \in \NN_{+}$ there exist $n_1, \dots, n_{\ell} \in \NN$
		with $n_1 + n_2 + \dots + n_{\ell} \geq n$, and $(n_{ij})_{ i = 1, \dots, 
		k\text;~j = 1, \dots, \ell }$ with 
		\begin{equation}\label{eq:centraldecomp} 
			\big( n_{1j}, \dots, n_{kj} \big) \in N_Q( n_j )
		\end{equation}
		for all $j = 1, \dots, \ell$; and there exists $\epsilon_0 > 0$ such that
		for all positive $\epsilon < \frac{\epsilon_0}4$ there are
		$f_{1,\epsilon}^*, \dots, f_{\ell,\epsilon}^* \in \FF_Q^*$ with the
		following properties:
		\begin{enumerate}[(a)]
			\item\label{itm:centralconda} For all $j,j' = 1,\dots, \ell$ with $j
				\neq j'$ it holds that $d_Q( f_{j,\epsilon}^*, f_{j',\epsilon}^* ) >
				\frac{\epsilon_0}2$.
			\item\label{itm:centralcondb} For all $i = 1, \dots, k$ and all $j =
				1, \dots, \ell$ it holds that $\card{ D }_{ q_i( \FFFF ) \cap q_i(
				B_\epsilon( f_{j,\epsilon}^* ) ) } = n_{ij}$\label{eq:centralapprox}
		\end{enumerate}
	\end{cond}

	An illustration of the situation of \cref{cond:central} can be found in
	\cref{fig:central}. We now prove both directions of the claimed equivalence.
		\begin{description}
			\item[$\Rightarrow$]
			First assume that $\card{ Q(D) }_{ \FFFF } \geq n > 0$ and suppose
			that $f_1, \dots, f_{\ell}$ are (pairwise distinct) facts from
			$\FFFF$ with the property $\card{ D }_{ q_i( f_j ) } > 0$ for at least
			one $i = 1, \dots, k$ for all $j = 1, \dots, \ell$. Such $f_1, \dots,
			f_{\ell}$, with $\ell \in \NN_{+}$, exist by requirement
			\labelcref{itm:central1}. Let $n_j \coloneqq \card{ Q(D) }_{ f_i }$
			for all $j = 1, \dots, \ell$.  Then $\sum_{ j = 1 }^{ \ell } n_j \geq
			n$, and, according to \labelcref{itm:central1}, $\big( n_{1j}, \dots,
			n_{kj} \big) \in N_Q( n_j )$, where $n_{ij} \coloneqq \card{ D }_{
			q_i( f_j ) }$ for all $i =1,\dots,k$ and $j = 1,\dots,\ell$. We choose
			$\epsilon_0 > 0$ such that
			\[
				\epsilon_0 < 
				\min\set[\big]{ d_Q( f_j, f ) \with 
					j=1, \dots, \ell, f \neq f_j, f \in Q(D) }
			\]
			and let $0 < \epsilon < \frac{\epsilon_0}4$. Because $\FF_Q^*$ is
			dense in $\FF_Q$, we can choose $f_{1,\epsilon}^*, \dots,
			f_{\ell,\epsilon}^* \in \FF_Q$ such that $d_Q( f_j, f_{j,\epsilon}^* )
			< \epsilon < \frac{\epsilon_0}4$ for all $j = 1, \dots, \ell$. Then
			from the triangle inequality it follows that
			\[
				d_Q( f_{j,\epsilon}^*, f_{j',\epsilon}^* ) 
				\geq 
				d_Q\big( f_j, f_{j'} \big) - 
				d_Q\big( f_{j,\epsilon}^*, f_j \big) -
				d_Q\big( f_{j'}, f_{j',\epsilon}^* \big)
				> 
				\epsilon_0 - 2 \tfrac{\epsilon_0}4 = \tfrac{\epsilon_0}2
				\text.
			\]
			In particular, since $\epsilon < \tfrac{\epsilon}4$, the balls
			$B_{\epsilon}( f_{j,\epsilon}^* )$ are pairwise disjoint and 
			$f_j \in B_{\epsilon}( f_{j,\epsilon}^* )$ for all $j = 1,\dots,\ell$.
			Moreover, by the choice of $\epsilon_0$, no fact $f$ from $Q(D)$ other
			than $f_1,\dots,f_\ell$ appears in the balls $B_{\epsilon}(
			f_{1,\epsilon}^* ), \dots B_{\epsilon}( f_{j,\epsilon}^* )$.
			Because $q_i$ is injective, it follows that $q_i( f ) \notin q_i(
			B_{\epsilon}( f_{j,\epsilon}^* ) )$ for all $f \neq f_j$ appearing in
			$Q(D)$, and all $i = 1, \dots, k$. This means that for all $i =
			1,\dots,k$ and $j =1, \dots, \ell$ it holds that
			\[
				\card{ D }_{ q_i( \FFFF ) \cap q_i( B_\epsilon( f_{j,\epsilon}^* )
				) } 
				= \card{ D }_{ q_i( f_j ) } 
				= n_{ij}\text.
			\]
			Thus, \cref{cond:central} holds.
		\item[$\Leftarrow$] Suppose that \cref{cond:central} holds and let $\ell 
			\in \NN_+$, $n_1 + \dots + n_{\ell} \geq n$ and $n_{ij}$ for
			$i=1,\dots,k$ and $j = 1,\dots,\ell$, and $\epsilon_0$ be given
			accordingly. Note that for all $\epsilon < \frac{\epsilon_0}4$, the
			balls $B_{\epsilon}( f_{j,\epsilon}^* )$ are pairwise disjoint by
			condition \labelcref{itm:centralconda}. Since $q_i$ is injective, the
			sets $q_i\big( B_{\epsilon}( f_{j,\epsilon}^* )) \big)$ are disjoint
			as well. By \labelcref{itm:centralcondb}, for all $i = 1, \dots, k$,
			and all $j = 1, \dots, \ell$ it holds that
			\[
				\card{ D }_{ q_i( \FFFF ) \cap 
				q_i( B_{\epsilon}( f_{j,\epsilon}^* ) ) } 
				= n_{ij}\text.
			\]
			If $\epsilon$ is small enough, then for all $n_{ij}$ with $n_{ij} > 0$
			it holds that there exists a single fact $f_{i,j,\epsilon} \in q_i(
			\FFFF ) \cap q_i(B_{\epsilon}( f_{j,\epsilon}^* ))$ with the property
			that $\card{ D }_{ f_{i,j,\epsilon} } = n_{ij}$. To see this, note 
			that if the fact count $n_{ij}$ were distributed over multiple facts,
			then by the injectivity of the functions $q_i$, these facts have
			pairwise different preimages under $q_i$. For small enough $\epsilon$,
			however, $B_{\epsilon}( f_{j,\epsilon}^* )$ can only contain one of
			these facts, yielding a contradiction with
			\labelcref{itm:centralcondb}. We now fix $\epsilon$ small enough such
			that the above holds, as well as an index $j = 1, \dots, \ell$.
			Moreover, we pick some $i(j) \in \set{ 1, \dots, k }$ with $n_{i(j),j}
			> 0$ and define $f_j \coloneqq q_{i(j)}^{-1}( f_{i(j),j,\epsilon} )$.
			As $q_{i(j)}$ is injective, it follows that $f_j \in \FFFF \cap
			B_{\epsilon}( f_{j,\epsilon}^* )$. Without loss of generality, we
			assume that $f_j \in B_{\epsilon'}( f_{j,\epsilon'}^* )$ for all $0 <
			\epsilon' < \epsilon$.\footnote{Generally, it could happen that $j$
			switches roles with some $j'$ where $(n_{1j'},\dots,n_{kj'}) = (
			n_{1j},\dots,n_{kj})$. If there are exactly $m$ indices $j$ for
			which the sequences $(n_{1j},\dots,n_{kj})$ coincide, the
			corresponding balls $B_{\epsilon}( f_{j,\epsilon}^* )$ for these
			indices have (at least) $m$ distinct accumulation points.} This means
			that $\bigcap_{ 0 < \epsilon' < \epsilon } B_{\epsilon'}(
			f_{j,\epsilon'}^* ) = \set{ f_j }$.

			We claim that for all $i = 1, \dots, k$ it then holds that
			\[
				\card{ D }_{ q_i( f_j ) } = n_{ij}\text.
			\]
			By choice, this already holds for $i = i(j)$. Suppose $i$ is an index
			with $n_{ij} = 0$. As $\card{ D }_{ q_i( \FFFF ) \cap
			q_i(B_{\epsilon}( f_{j,\epsilon}^* )) } = n_{ij} = 0$ and $f_j \in
			B_{\epsilon}(f_{j,\epsilon}^*)$, it follows that $\card{D}_{q_i(f_j)}
			= 0$. Now let $i$ be an index different from $i(j)$ such that
			$n_{ij}>0$. We have argued before that then there exists a single fact
			$f_{i,j,\epsilon} \in q_i( \FFFF ) \cap q_i( B_{\epsilon}(
			f_{j,\epsilon}^* )) = n_{ij}$. Assume that $f_{i,j,\epsilon} 
			\neq q_i(f_j)$. Then $q_i^{-1}( f_{i,j,\epsilon} ) \neq f_j$ and
			therefore, there is $\epsilon'< \epsilon$ such that $q_i^{-1}(
			f_{i,j,\epsilon'} ) \notin B_{\epsilon}( f_{j,\epsilon'}^* )$, leading
			to a contradiction with \labelcref{itm:centralcondb}. Thus,
			$f_{i,j,\epsilon} = q_i(f_j)$ and $\card{D}_{ q_i(f_j) } = 
			\card{D}_{ f_{i,j,\epsilon} } = n_{ij}$. From property
			\labelcref{itm:central1}, it follows that $\card{Q(D)}_{ f_j } = n_j$.
			Thus, since $n_1+\dots+n_\ell \geq n$ and since the facts $f_j \in
			\FFFF$ are pairwise distinct, it follows that $\card{Q(D)}_{\FFFF} \geq
			n$.
	\end{description}

	To conclude the proof, we argue that \cref{cond:central} can be used to
	express $\card{Q(D)}_{ \FFFF } \geq n$ as a countable combination of 
	counting events. Note that because $\QQ_+$ is dense in $\RR_+$, the 
	equivalence of $\card{ Q(D) }_{ \FFFF } \geq n$ and \cref{cond:central}
	still holds, if the numbers $\epsilon_0$ and $\epsilon$ are additionally 
	required to be rational. Also, the fact sets in \labelcref{eq:centralapprox}
	are measurable in $\FFF_i$: the open balls are certainly measurable, and as
	$q_i$ is injective and continuous, $q_i$ maps measurable sets to measurable
	sets by \cref{fac:SBSinjmeasimg}. Then the set of database instances $D \in
	\DB$ with \cref{cond:central} is of the shape
	\[
		\bigcup_{\ell, (n_i), (n_{ij})}
		\bigcup_{\epsilon_0}\bigcap_{\epsilon}\bigcup_{ (f_{j,\epsilon}^*) }
		\bigcap_{i,j}
		\set[\big]{ D \in \DB \with \card{ D }_{ q_i( \FFFF ) \cap 
		q_i(B_{\epsilon}(f_{j,\epsilon}^*)) } = n_{ij} }\text,
	\]
	with the indices ranging as in \cref{cond:central} and $\epsilon$, 
	$\epsilon_0$ additionally restricted to $\QQ$.
\end{proof}

\begin{figure}[H]
		\centering
		\begin{tikzpicture}
			\node[draw,minimum width=3cm,minimum height=4.85cm,rectangle,rounded
				corners=.5cm, label={below:$\FF_Q$}] (Q) at (0,0) {};
			\node[draw,minimum width=3cm,minimum height=4.85cm,rectangle,rounded
				corners=.5cm, right=3.25cm of Q,label={below:$\FF_{R_1}$}] (R1) {};
			\node[draw,minimum width=3cm,minimum height=4.85cm,rectangle,rounded
				corners=.5cm, right=2cm of R1,label={below:$\FF_{R_2}$}] (R2) {};

			\def\origF{plot coordinates {(-1.25,1.2) (-.2,1.8) (1.3,2) (1.1,-1.0) 
				(-0.8,-1.3)}}

			\fill[shade1,smooth cycle] \origF;
			\coordinate (ff) at (-.4,.9);	
			\coordinate (ff') at (-.3,-1.7);
			\begin{scope}
				\path[clip,smooth cycle] \origF;
				\fill[shade2] (ff) circle (.6cm);
				\fill[shade2] (ff') circle (.6cm);
			\end{scope}
			\draw[smooth cycle] \origF;
			\node at (0.7,-1.7) {$\FFFF$};

			\node[point,label={[yshift=.4ex]below right:$f_j$}] (f) at (-0.2,0.5) {};
			\node[shade3,point,label={left:$f_{j,\epsilon}^*$}] at (ff) {};
			\draw[shade3] (ff) circle (.6cm);
			\node[above right=1.3ex and 1.4ex of ff] {$B_{\epsilon}(f_{j,\epsilon}^*)$};
			\draw[shade3,-stealth',dashed] (ff) to (f);

			\node[point,label={[yshift=.4ex]above:$f_{j'}$}] (f') at (-0.5,-1.3) {};
			\node[shade3,point,label={[xshift=-.8ex,yshift=1ex]below right:$f_{j',\epsilon}^*$}] at (ff') {};
			\draw[shade3] (ff') circle (.6cm);
			\draw[shade3,-stealth',dashed] (ff') to (f');

			\begin{scope}[xshift=6.25cm]
				\def\qFi{plot coordinates {(-1,0) (.2,.6) (.9,0) (.9,-.9) 
					(-.8,-1.4)}}
				\fill[shade1,smooth cycle] \qFi;
				\coordinate (gg) at (-.6,-.4);
				\begin{scope}
					\path[clip,smooth cycle] \qFi;
					\fill[shade2] (gg) circle (.8cm);
				\end{scope}
				\draw[smooth cycle] \qFi;
				\node at (0.1,-1.7) {$q_1(\FFFF)$};
				\node[point] (g) at (-0.6,-1) {};
				\node[shade3,point] at (gg) {};
				\draw[shade3] (gg) circle (.8cm);

				\draw[shade3,-stealth',dashed] (gg) to (g);
			\end{scope}

			\begin{scope}[xshift=11.25cm]
				\def\qFii{plot coordinates {(-1.2,2) (0,1.8) (1,1.6) (1.1,-.2) (-.8,.6)}}
				\fill[shade1,smooth cycle] \qFii;
				\coordinate (hh) at (.1,1.2);
				\begin{scope}
					\path[clip,smooth cycle] \qFii;
					\fill[shade2] (hh) circle (.8cm);
				\end{scope}
				\draw[smooth cycle] \qFii;
				\node at (0.5,-.6) {$q_2(\FFFF)$};
				\node[point] (h) at (-.2,.7) {};
				\node[shade3,point] at (hh) {};
				\draw[shade3] (hh) circle (.8cm);
				\draw[shade3,-stealth',dashed] (hh) to (h);
			\end{scope}

			\draw[shade3,-stealth'] (ff) to[bend left=10] (gg);
			\draw[shade3,-stealth'] (ff) to[bend left=10] (hh);

			\draw[-stealth'] (f) to[bend left=10] node[above,pos=.65] {$q_1$} (g);
			\draw[-stealth'] (f) to[bend left=10] node[above,pos=.825] {$q_2$} (h);
		\end{tikzpicture}
		\caption{Illustration of \cref{cond:central} in the case $k = 2$. Note 
			that $f_{j,\epsilon}^*$ itself need not necessarily lie in $\FFFF$.
			The dashed arrows indicate the intended convergence of
			$f_{j,\epsilon}^*$ to the fact $f_j$ as $\epsilon \to 0$.
		}\label{fig:central}
	\end{figure}

Incidentally, \cref{thm:mapping} is a special case of \cref{lem:central} for
$k = 1$ and with $N_Q( n ) \coloneqq \set{ n }$ for all $n \in \NN_+$.

\subsection{Coarse Preimages}\label{ssec:coarse}

For $D \in \DB$, and every relation symbol $R \in \tau$, we let
\[
	\TTTT_R( D ) 
	\coloneqq \set[\big]{ t \in \TT_R \with \card{ D }_{ R(t) } > 0 }
\]
denote the set of $R$-tuples in $D$. Moreover, we let 
\[
	d_R( D ) \coloneqq 
	\begin{cases}
		\min\set{ d_R( t, t' ) \with t, t' \in \TTTT_R(D) \text{ with } t \neq t' }
		& \text{if } \card[\big]{ \TTTT_R(D) } \geq 2 \text{ and}\\
		\infty
		& \text{otherwise.}
	\end{cases}
\]
That is, $d_R( D )$ is the smallest distance between any two $R$-tuples of $D$
(or $\infty$, if $D$ contains at most one $R$-fact).

\begin{defi}
	Let $\epsilon > 0$. An instance $D \in \DB$ is called 
	\emph{$\epsilon$-coarse} if for all $R \in \tau$ it holds that $d_R( D ) > 
	\epsilon$. We denote the set of $\epsilon$-coarse instances by
	$\DB\rvert_\epsilon$. 
\end{defi}

Unfolding the definition, all instances in $\DB\rvert_\epsilon$ have the 
property that their tuples (per relation) are sufficiently far apart with
respect to the metric on the respective space of tuples. 

\begin{exa}\label[exa]{exa:runningcoarse}
	Consider our example of temperature recordings, but for simplicity (in order
	for not having to discuss the metric on the product space), assume that
	there is only one relation, with a single attribute for (real-valued)
	temperature recordings. Then a database instance is $\epsilon$-coarse
	precisely if the temperatures occuring in its instances differ by more than
	$\epsilon$ between distinct facts.
\end{exa}

The following lemma states that the set of $\epsilon$-coarse instances is
measurable for any $\epsilon > 0$.

\begin{lem}\label[lem]{lem:coarse}
	It holds that $\DB\rvert_\epsilon \in \DDB$ for all $\epsilon > 0$.
\end{lem}

\begin{proof}
	An instance $D \in \DB$ is \emph{not} $\epsilon$-coarse, if for some $R \in 
	\tau$, there exist $t_1, t_2 \in \TTTT_R( D )$ such that $t_1 \neq t_2$ and
	$d_R( t_1, t_2 ) \leq \epsilon$. We claim that this is the case if and only
	if $D$ satisfies the following condition.

	\begin{cond}\label[cond]{cond:coarseequiv}
		There are $k_1, k_2 \in \NN_+$ and $\epsilon_0 > 0$ such that for all $r 
		\in \big( 0, \epsilon_0 \big)$ there are $t_{1,r}^*, t_{2,r}^* \in
		\TT_R^*$ with $\epsilon_0 < d_R( t_{1,r}^*, t_{2,r}^* ) < \epsilon + 2r$
		such that 
		\begin{enumerate}
			\item $\card{ D }_{ R( B_r(t_{1,r}^*) ) } = k_1$ and
			\item $\card{ D }_{ R( B_r(t_{2,r}^*) ) } = k_2$.
		\end{enumerate}
	\end{cond}

	The condition basically states that we can find approximations $t_{1,r}^*$
	and $t_{2,r}^*$ of two tuples $t_1$ and $t_2$ in $\TTTT_R(D)$ that witness
	that $d_R(t_1,t_2)$ is too small. In particular, the role of $\epsilon_0$ in 
	\cref{cond:coarseequiv} is to guarantee that $t_{1,r}^*$ and $t_{2,r}^*$ do
	not end up approximating the same tuple. We proceed to show the following
	equivalence:
	\[
		\text{there are } t_1,t_2 \in \TTTT_R(D), t_1\neq t_2 \text{ with }
		d_R(t_1,t_2) \leq \epsilon \quad\Leftrightarrow\quad 
		D \text{ satisfies \cref{cond:coarseequiv}}\text.
	\]
	To simplify notation, we let $d \coloneqq d_R$ denote our metric on $\TT_R$,
	and $d( D ) \coloneqq d_R( D )$.
	\begin{description}
		\item[$\Rightarrow$] 
			Let $t_1,t_2$ be two distinct tuples in $\TTTT_R(D)$ of minimal 
			distance, in particular satisfying $0 < d( t_1, t_2 ) \leq \epsilon$.
			Let $k_1 \coloneqq \card{ D }_{ R(t_1) }$ and $k_2 \coloneqq \card{ D
			}_{ R(t_2) }$. We fix some $\epsilon_0 > 0$ that satisfies $\epsilon_0 
			< \tfrac13 d(t_1,t_2)$. Since $\TT_R^*$ is dense in $\TT_R$, for every
			$r > 0$
			(in particular for $r < \epsilon_0$) there exist $t_{1,r}^*, t_{2,r}^*
			\in \TT_R^*$ with $d( t_{1,r}^*, t_1 ), d( t_{2,r}^*, t_2 ) < r$.
			Since $d( t_1, t_2 ) > 3\epsilon_0$ and $r < \epsilon_0$, $r > 0$,
			using the triangle inquality, it follows that 	
			\[
				d( t_{1,r}^*, t_2 ) 
				\geq d( t_1, t_2 ) - d( t_1, t_{1,r}^* ) 
				> 3\epsilon_0 - r
				> \epsilon_0 
				> r
			\]
			and similarly that $d( t_{2,r}^*, t_1 ) > \epsilon_0 > r$. Because
			$t_1$ and $t_2$ were chosen with minimal distance, this entails
			\(
				\card{ D }_{ R( B_r(t_{1,r}^*) ) } 
				= \card{ D }_{ R( t_1 ) } = k_1
			\) and \(
				\card{ D }_{ R( B_r(t_{2,r}^*) ) } 
				= \card{ D }_{ R( t_2 ) } = k_2
			\).
			Note that, again using the triangle inequality, it holds that
			\[
				d( t_{1,r}^*, t_{2,r}^* ) 
				\leq d( t_{1,r}^*, t_1 ) + d( t_1, t_2 ) + d( t_2, t_{2,r}^*)
				< \epsilon + 2r
			\]
			and, moreover, 
			\[
				d( t_{1,r}^*, t_{2,r}^* )
				\geq d( t_1, t_2 ) - d( t_1, t_1^* ) - d( t_2, t_{2,r}^* )
				> \epsilon_0
				\text.
			\]
			Together, $D$ satisfies \cref{cond:coarseequiv}.
		\item[$\Leftarrow$] Now suppose \cref{cond:coarseequiv} holds for some
			$k_1, k_2 \in \NN_{+}$, and some $\epsilon_0 > 0$. By 
			\cref{cond:coarseequiv}, for all positive $r < \epsilon_0$ there exist
			$t_{1,r}^*, t_{2,r}^* \in \TT_R^*$ with $\epsilon_0 < d( t_{1,r}^*,
			t_{2,r}^* ) < \epsilon + 2r$ such that
			\( 
				\card{ D }_{ R( B_r(t_{1,r}^*) ) } = k_1 > 0
			\) and \(
				\card{ D }_{ R( B_r(t_{2,r}^*) ) } = k_2 > 0
			\).
			Note that this implies $\epsilon_0 < \epsilon$.	

			Now for all $r < \tfrac13 \min\set[\big]{ d_R( D ), \epsilon_0 }$, the
			balls $B_r( t_{1,r}^* )$ and $B_r( t_{2,r}^* )$ are disjoint, and
			contain at exactly one tuple from $\TTTT_R( D )$ each, say $\TTTT_R(D)
			\cap B_r( t_{1,r}^* ) = \set{ t_1 }$ and $\TTTT_R(D) \cap B_r( t_{2,r}^*
			) = \set{ t_2 }$. Thus,
			\(
				\card{ D }_{ R(t_1) } = \card{ D }_{ R( B_r(t_{1,r}^*) ) } = k_1 > 0
			\) and \(
				\card{ D }_{ R(t_2) } = \card{ D }_{ R( B_r(t_{2,r}^*) ) } = k_2 > 0
			\).
			Moreover,
			\begin{equation}\label{eq:coarselimit}
				d( t_1, t_2 ) \leq d( t_1, t_{1,r}^* ) + d( t_{1,r}^*, t_{2,r}^* ) + 
				d( t_{2,r}^*, t_2 ) < \epsilon + 4r\text.
			\end{equation}
			In particular, \cref{cond:coarseequiv} implies that for all $r > 0$ 
			there exist distinct $t_1, t_2 \in \TTTT_R( D )$ with $d( t_1, t_2 ) < 
			\epsilon+4r$. As $D$ is finite, letting $r \to 0$ implies that there 
			exist distinct $t_1, t_2 \in \TTTT_{R}( D )$ with $d( t_1, t_2 ) \leq
			\epsilon$. 
	\end{description}
	Because $\QQ$ is dense in $\RR$, the equivalence still holds, if the numbers
	$\epsilon_0$ and $r$ are restricted to $\QQ$. Then,
	\[
		\big(\DB\rvert_{\epsilon}\big)^{ \compl } = 
		\bigcup_{ R } \bigcup_{ \epsilon_0 } \bigcap_{ r } 
		\bigcup_{ t_{1,r}^*, t_{2,r}^* } 
		\set[\big]{ 
			D \in \DB \with 
			\card{ D }_{ R(B_r(t_{1,r}^*)) } = k_1 \text{ and } 
			\card{ D }_{ R(B_r(t_{2,r}^*)) } = k_2
		}
		\in \DDB
	\]
	(with the indices ranging as in our equivalence, numbers being restricted to
	rationals). Thus, $\DB\rvert_\epsilon \in \DDB$.
\end{proof}

Note that for every database instance $D$, there exists an $\epsilon > 0$,
small enough, such that $D \in \DB\rvert_\epsilon$, because $D$ is finite. This
means that $\bigcup_{ \epsilon > 0 } \DB\rvert_{ \epsilon }$ covers $\DB$ (even
if the union is taken only over rational $\epsilon$).

\begin{cor}\label[cor]{cor:coarsepreimage}
	If for all $\epsilon > 0$, all $\FFFF \in \FFF_Q$, and all $n \in \NN$ it
	holds that
	\[
		\set[\big]{ 
			D \in \DB\rvert_\epsilon \with \card{ Q( D ) }_{ \FFFF } = n
		} \in \DDB\text,
	\]
	then $Q$ is measurable.
\end{cor}

\begin{proof}
	This follows directly from
	\[
	\set[\big]{ D \in \DB \with \card{ Q(D) }_{ \FFFF } = n }
		= \bigcup\nolimits_{ \epsilon \in \QQ_{+} } 
		\set[\big]{ D \in \DB\rvert_\epsilon \with \card{ Q(D) }_{ \FFFF } = n }.
		\qedhere
	\]
\end{proof}

\Cref{cor:coarsepreimage} can be used to leverage the finiteness of our 
instances. In an $\epsilon$-coarse instance, we can approximate sets of facts
by simpler sets of facts as long as these approximations are sufficiently fine.
For example, in the context of \cref{exa:runningcoarse}, in order to prove
measurability of a query, it suffices to prove the measurability with respect
to preimages where the temperature recordings are \enquote{far apart}.

\section{Relational Algebra}\label{sec:ra}
As motivated in \cref{ssec:pws}, we investigate the measurability of relational
algebra queries in our model. The concrete relational algebra for bags that we
use here is basically the (unnested version of the) algebra that was introduced
in~\cite{Dayal+1982} and investigated, respectively extended, and surveyed
in~\cite{Albert1991,GrumbachMilo1996,Grumbach+1996}. It is called
$\mathsf{BALG^1}$ (with superscript $1$) in~\cite{GrumbachMilo1996}. 

We do not introduce nesting as it would yield yet another layer of abstraction
and complexity to the spaces we investigate, although by the properties that
such spaces exhibit, we have strong reason to believe that there is no 
technical obstruction in allowing spaces of finite bags as attribute domains.
It is unclear however, whether this extends to PDBs with unbounded nesting
depth.

The operations we consider are shown in the \cref{tab:balg} below. As seen
in~\cite{Albert1991,GrumbachMilo1996,Grumbach+1996}, there is some redundancy
within this set of operations that will be addressed later. A particular
motivation for choosing this particular algebra is that possible worlds
semantics are usually built on top of set semantics and these operations
naturally extend the common behavior of relation algebra queries to bags. This
is quite similar to the original motivation of~\cite{Dayal+1982}
and~\cite{Albert1991} regarding their choice of operations.

\begin{table}[H]
	\caption{The operators of $\mathsf{BALG}^1$ considered in this paper.}\label{tab:balg}
	\begin{tabular}{lll}
		\toprule
		\textit{Base Queries}	
			& Constructors			& $Q = \bag{}$ and $Q = \bag{ R(t_0) }$\\
			& Extractors			& $Q = R$\\
			& Renaming				& $Q = \bagren{ A \to B }{ R }$\\
		\midrule
		\textit{Basic Bag Operations}	
			& Additive Union		& $Q = \bagadd{ R }{ S }$\\
			& Max Union				& $Q = \bagmax{ R }{ S }$\\
			& Intersection			& $Q = \bagmin{ R }{ S }$\\
			& Difference			& $Q = \bagdiff{ R }{ S }$\\
			& Deduplication		& $Q = \bagded{ R }$\\
		\midrule
		\textit{SPJ-Operations}	
			& Selection 			& $Q = \bagsel{ (A_1,\dots,A_k) \in \BBBB }{ R }$\\
			& Projection			& $Q = \bagproj{ A_1,\dots,A_k }{ R }$\\
			& Cross Product		& $Q = \bagprod{ R }{ S }$\\
		\bottomrule
	\end{tabular}
\end{table}

Since compositions of measurable functions are measurable, it suffices to show
the measurability of the operators from \cref{tab:balg}, and the measurability
of compound queries follows by structural induction. 

Therefore, by investigating the measurability of the operators from 
\cref{tab:balg} we will show the following main result of this section.

\begin{thm}\label{thm:balg}
	All queries expressible in the bag algebra $\mathsf{BALG^1}$ are measurable.
\end{thm}

\subsection{Base Queries}
The base queries (\cref{tab:base}) are easily seen to be measurable.

\begin{table}[H]
	\centering
	\caption{Base Queries.}\label{tab:base}
	\begin{tabular}{ ll }
		\toprule%
		\textbf{Query} 
		& \textbf{Semantics} (for all $D \in \DB$, $t \in \TT_Q$)									
		\\%
		\midrule%
		$Q = \bag{}$				
		& $\card{ Q(D) }_{ R_Q(t) } \coloneqq 0$
		\\%
		$Q = \bag{ R_Q( t_0 ) }$
		& $\card{ Q(D) }_{ R_Q(t) } \coloneqq 1$ if $t = t_0$, and $0$ otherwise
		\\%
		$Q = R$
		& $\card{ Q(D) }_{ R_Q(t) } \coloneqq \card{ D }_{ R( t ) }$
		\\%
		\bottomrule%
	\end{tabular}
\end{table}

\begin{lem}\label[lem]{lem:base}
	The following queries are measurable:
	\begin{enumerate}
		\item $Q = \bag{}$,
		\item $Q = \bag{ f }$ for all $f \in \FF_Q$, and
		\item $Q = R$ for all $R \in \tau$.
	\end{enumerate}
\end{lem}

\begin{proof}
	\begin{enumerate}
		\item Let $\DDDD \in \DDB_Q$. If $\bag{ } \in \DDDD$, then 
			$Q^{-1}( \DDDD ) = \DB \in \DDB$. Otherwise, if $\bag{ } \notin
			\DDDD$, then $Q^{-1}( \DDDD ) = \emptyset \in \DDB$.
		\item This can be shown the same way.
		\item Let $Q = R$, let $R_Q( \TTTT ) \in \FFF_Q$ and let $n \in \NN$. 
			Then $R( \TTTT ) \in \FFF_R \subseteq \FFF$. For all $D \in \DB$, it
			holds that
			\[
				\card{ Q( D ) }_{ R_Q( \TTTT ) } = n
				\iff
				\card{ D }_{ R( \TTTT ) } = n\text,
			\]
			and the latter is measurable in $\DDB$. Thus, $Q$ is measurable.
			\qedhere
	\end{enumerate}	
\end{proof}

Note that there is nothing to show for the renaming query because it leaves
all tuples themselves untouched. In the subsequent subsections, we deal with
the remainder of \cref{tab:balg}.

\subsection{Basic Bag Operations}

We now treat the basic bag operations (\cref{tab:basicbags}). Assume that $R, S
\in \tau$ are relation symbols of the same sort. From now on, we will only
consider the case where $R$ and $S$ are distinct, as for the case $R=S$,
measurability is trivial.

\begin{table}[H]
	\centering
	\caption{Basic Bag Operations.}\label{tab:basicbags}
	\begin{tabular}{ ll }
		\toprule%
		\textbf{Query}
		& \textbf{Semantics}	(for all $D \in \DB$, $t \in \TT_Q$)									
		\\%
		\midrule%
		$Q = \bagadd{R}{S}$				
		& $\card{ Q(D) }_{ R_Q(t) } \coloneqq 
			\card{ D }_{ R(t) } + \card{ D }_{ S(t) }$
		\\%
		$Q = \bagdiff{R}{S}$
		& $\card{ Q(D) }_{ R_Q(t) } \coloneqq 
			\max\set[\big]{ 0, \card{ D }_{ R(t) } - \card{ D }_{ S(t) } }$
		\\%
		$Q = \bagmax{R}{S}$
		& $\card{ Q(D) }_{ R_Q(t) } \coloneqq 
			\max\set[\big]{ \card{ D }_{ R(t) }, \card{ D }_{ S(t) } }$
		\\%
		$Q = \bagmin{R}{S}$
		& $\card{ Q(D) }_{ R_Q(t) } \coloneqq 
			\min\set[\big]{ \card{ D }_{ R(t) }, \card{ D }_{ S(t) } }$
		\\%
		\bottomrule%
	\end{tabular}
\end{table}

\begin{lem}\label{lem:basicbagops}
	The following queries are measurable:
	\begin{enumerate}
		\item $Q = \bagadd{R}{S}$.\label[itm]{itm:bagops1}
		\item $Q = \bagdiff{R}{S}$.\label[itm]{itm:bagops2}
		\item $Q = \bagmax{R}{S}$.\label[itm]{itm:bagops3}
		\item $Q = \bagmin{R}{S}$.\label[itm]{itm:bagops4}
	\end{enumerate}
\end{lem}

\begin{proof}
	As $\cup$ and $\cap$ can be expressed by compositions of $\uplus$ and $-$
	\cite{Albert1991}, it suffices to show \labelcref{itm:bagops1,itm:bagops2}.
	\begin{enumerate}
		\item Consider the functions $g \from \FF_Q \to \FF_R$ and $h \from \FF_Q 
			\to \FF_S$ with
			\[
				g\big( R_Q( t ) \big) = R( t )
				\centertext{and}
				h\big( R_Q( t ) \big) = S( t )
			\]
			for all $t \in \TT_Q = \TT_R = \TT_S$. Clearly, both functions are 
			injective and continuous. For all $k \in \NN$ we define 
			\[
				N_Q( n ) \coloneqq \set{ (n_1, n_2) \in \NN^2 \with n_1 + n_2 = n }
				\text.
			\]
			Then $\card{ Q(D) }_f = n$ if and only if $\big( \card{ D }_{ g(f) },
			\card{ D }_{ h(f) } \big) \in N_Q(n)$. Thus, $Q$ is measurable by 
			\cref{lem:central}.
		\item This works analogously to part \labelcref{itm:bagops1}, with 
			\[
				N_Q(n) \coloneqq \set{ (n_1,n_2) \in \NN^2 \with \max\set{ n_1-n_2, 0
				} = n}
				\text.\qedhere
			\]
	\end{enumerate}
\end{proof}

\subsection{Set Semantics}

The deduplication operator (\cref{tab:dedupe}) maps bag instances to their 
underlying set instances.

\begin{table}[H]
	\centering
	\caption{Deduplication.}\label{tab:dedupe}
	\begin{tabular}{ ll }
		\toprule%
		\textbf{Query}
		& \textbf{Semantics}	(for all $D \in \DB$, $t \in \TT_Q$)									
		\\%
		\midrule%
		$Q = \bagded{R}$				
		& $\card{ Q(D) }_{ R_Q(t) } \coloneqq 1$ if $\card{ D }_{ R(t) } > 0$,
			and $0$ otherwise
		\\%
		\bottomrule%
	\end{tabular}
\end{table}

\begin{lem}\label{lem:dedupe}
	The deduplication query $Q = \bagded{ R }$ is measurable for all $R \in
	\tau$.
\end{lem}

\begin{proof}
	We apply \cref{lem:central} for $r = 1$ and the function $q_1 = q \from
	\FF_Q \to \FF_R$ defined by $q( R_Q(x) ) = R(x)$. Then for all $D
	\in \DB$ and all $f \in \FF_Q$, it holds that
	\[
		\card{ Q(D) }_{ f } = \card{ D }_{ q( f ) }\text.
	\]
	Thus, we let $N_Q( k ) = \set{ 0 }$ if $k = 0$ and $N_Q( k ) = \NN \setminus
	\set{ 0 }$ otherwise. Then $r$, $q$ and $N_Q$ satisfy the requirements of
	\cref{lem:central}, so $Q$ is measurable.
\end{proof}

Having the deduplication query measurable means that standard PDBs support set
semantics.

\begin{rem}
	The function associated with the deduplication query is countable-to-one
	(preimage of a single instance in the result is a countable collection of
	instances) and measurable by the lemma above. This can be used to infer that
	the space of set instances is standard Borel using~\cite[Theorem
	4.12.4]{Srivastava1998}. This means that we could also completely restrict
	our setting to set instances without introducing new measurability problems.
	In general however, it can still be mathematically more conventient to use
	the full measurable space that was defined in \cref{ssec:pdbfpp}, even in a
	set semantics setting. The point processes defining PDBs should then just be
	\enquote{simple} in the sense of \cref{def:fpp}, i.\,e.\ have the probability
	$0$ of containing duplicate facts.
\end{rem}

\subsection{Selection and Projection} 

In this section, we investigate the selection and projection operators
\cref{tab:selproj}. First, we note that reordering the attributes in the sort
of a relation yields a measurable query.

\begin{lem}\label[lem]{lem:reordering}
	Let $\sort( R ) = (A_1, \dots, A_r)$ and let $\beta$ be a permutation of 
	$\set{1, \dots , r}$. For $f = R( a_1, \dots, a_r )$ let $\beta( f ) =
	R_Q( a_{ \beta(1) }, \dots, a_{ \beta(r) } )$, and let $Q( D ) \coloneqq 
	\set{ \beta(f) \with f \in D }$. Then $Q$ is measurable.
\end{lem}

\begin{proof}
	This directly follows from the mapping theorem (\cref{thm:mapping}), because
	$\beta \from \FF_R \to \FF_Q$ is measurable.
\end{proof}

\Cref{lem:reordering} is helpful for restructuring relations into a more
convenient shape to work with later. Semantically, it is a special case of a
projection query.

\bigskip

Let $R \in \tau$ be a relation symbol with $\ar( R ) = r$, and let $A_1, \dots,
A_k$ be pairwise distinct attributes appearing in $\sort(R)$ where $0 < k \leq
r$. By \cref{lem:reordering}, wlog. we assume that $\sort( R ) = ( A_1, \dots,
A_r )$.

\begin{table}[H]
	\centering
	\caption{Selection and Projection.}\label{tab:selproj}
	\begin{tabular}{ ll }
		\toprule%
		\textbf{Query}
		& \textbf{Semantics}	(for all $D \in \DB$, $t \in \TT_Q$)									
		\\%
		\midrule%
		$Q = \bagsel{(A_1, \dots, A_k) \in \BBBB}{R}$				
		& $\card{ Q(D) }_{ R_Q(t) } \coloneqq 
			\card{ D }_{ R(t) }$ if $t[A_1,\dots,A_k] \in \BBBB$, and $0$ 
			otherwise
		\\%
		$Q = \bagproj{A_1,\dots,A_k}{R}$
		& $\card{ Q(D) }_{ R_Q(t) } \coloneqq 
		\sum_{t' \with t'[A_1,\dots,A_k] = t} \card{ D }_{ R(t') }$
		\\%
		\bottomrule%
	\end{tabular}
\end{table}

\begin{lem}\label[lem]{lem:selection}
	Let $\BBBB \in \bigotimes_{ i = 1 }^{ k } \AAA_i$. Then the query 
	$Q = \bagsel{ (A_1, \dots, A_k) \in \BBBB }{ R }$ is measurable.
\end{lem}

\begin{proof}
	Fix $\FFFF \in \FFF_Q$, say $\FFFF = R_Q( \TTTT )$ with $\TTTT \in
	\bigotimes_{ i = 1 }^{ k } \AAA_i$, and let $n \in \NN$. Define
	\begin{equation}\label{eq:selFFFF-1}
		\FFFF^{ -1 }
		\coloneqq R( \TTTT ) \cap R ( \BBBB \times 
			\AA_{ k + 1 } \times \dotsc \times \AA_{ r }
		\big)
	\end{equation}
	Then $\FFFF^{ -1 } \in \FFF_R \subseteq \FFF$. It holds that
	\[
		\card{ Q(D) }_{ \FFFF } = n
		\iff
		\card{ D }_{ \FFFF^{-1} } = n\text.
	\]
	Thus, $Q$ is measurable.
\end{proof}

\begin{exa}
	Suppose $A_1, A_2 \in \sort( R )$ with $\AA_1 = \AA_2 = \RR$. The sets
	\begin{align*}
		\TTTT_{ = } \coloneqq \set{ ( x, y ) \in \RR^2 \with x = y }
		&&\text{and}&&
		\TTTT_{ < } \coloneqq \set{ ( x, y ) \in \RR^2 \with x < y }
		\intertext{are Borel in $\RR^2$. Thus, the queries}
		\bagsel{ A_1 = A_2 }{ R } \coloneqq \bagsel{ (A_1, A_2) \in \TTTT_= }{ R }
		&&\text{and}&&
		\bagsel{ A_1 < A_2 }{ R } \coloneqq \bagsel{ (A_1, A_2) \in \TTTT_< }{ R }
	\end{align*}
	are measurable by \cref{lem:selection}. In particular, in the context of our 
	running example, we can do selections based on comparing temperatures.
\end{exa}

\begin{lem}\label[lem]{lem:projection}
	The query $Q = \bagproj{ A_1, \dots, A_k }{ R }$ is measurable.
\end{lem}

\begin{proof}
	Fix $\FFFF \in \FFF_Q$, say $\FFFF = R_Q( \TTTT )$ with $\TTTT \in
	\bigotimes_{ i = 1 }^{ k } \AAA_i$, and let $n \in \NN$. Again, define
	\begin{equation}\label{eq:projFFFF-1}
		\FFFF^{-1} \coloneqq 
		R\big( \TTTT \times \AA_{ k+1 } \times \dotsc \times \AA_r \big)
		\text.
	\end{equation}
	Then $\FFFF^{ -1 } \in \FFF_R$ and it holds that
	\[
		\card{ Q(D) }_{ \FFFF } = n 
		\iff
		\card{ D }_{ \FFFF^{-1} } = n
		\text,
	\]
	and, hence, $Q$ is measurable.
\end{proof}

\begin{rem}
	Above, we provided direct proofs of \cref{lem:selection} and 
	\cref{lem:projection}.  Alternatively, they follow from \cref{thm:mapping}
	using the function
	\begin{align*}
		q \from \set[\big]{ R(t) \in \FF_R \with t[A_1,\dots,A_k] \in \BBBB } \to 
		\FF_Q
		&&\text{with}&&
		q\big( R(t) \big) &= R_Q(t)
	\intertext{for selection and the function}
		q \from \FF_R \to \FF_Q
		&&\text{with}&&
		q\big( R(t) \big) &= R_Q\big( t[A_1, \dots, A_k] \big)
	\end{align*}
	for projection. A closer look reveals that the sets $\FFFF^{-1}$ of 
	\labelcref{eq:selFFFF-1,eq:projFFFF-1} are really the preimages of the 
	respective function $q$ that we know from \cref{thm:mapping}.
\end{rem}

\subsection{Products}
Let $R, S \in \tau$ be relation symbols. In this section, we treat the cross
product $Q = R \times S$ (\cref{tab:cross}). (For the discussions here, it
does not matter whether $R=S$.)

\begin{table}[H]
	\centering
	\caption{Cross Products.}\label{tab:cross}
	\begin{tabular}{ ll }
		\toprule%
		\textbf{Query}
		& \textbf{Semantics}	(for all $D \in \DB$, $t \in \TT_R$, $t' \in \TT_S$)
		\\%
		\midrule%
		$Q = \bagprod{R}{S}$				
		& $\card{ Q(D) }_{ R_Q(t,t') } \coloneqq 
			\card{ D }_{ R(t) } \cdot \card{ D }_{ S(t') }$
		\\%
		\bottomrule%
	\end{tabular}
\end{table}

\begin{lem}\label[lem]{lem:crossprodrectangles}
	Let $\FFFF \in \FFF_Q$ such that $\FFFF = R_Q( \TTTT_R \times \TTTT_S )$
	with $\TTTT_R \in \TTT_R$ and $\TTTT_S \in \TTT_S$. Then
	\[
		\set{ D \in \DB \with \card{ Q(D) }_{ \FFFF } = n } \in \DDB
	\]
	for all $n \in \NN$.
\end{lem}

\begin{proof}
Let $\FFFF \in \FFF_Q$, say $\FFFF = R_Q( \TTTT )$ with $\TTTT \in \TT_Q$. If
$\TTTT = \TTTT_R \times \TTTT_S$ for some $\TTTT_R \in \TT_R$ and $\TTTT_S \in
\TT_S$, then for all $D \in \DD$, it holds that 
\begin{equation}\label{eq:proddecomp}
	\card{ Q( D ) }_{ \FFFF } = 
	\card{ D }_{ \FFFF_R } \cdot \card{ D }_{ \FFFF_S }
	\text,
\end{equation}
where $\FFFF_R = R( \TTTT_R )$ and $\FFFF_S = S( \TTTT_S )$. Thus,
\[
	\set{ D \in \DB \with \card{ Q(D) }_{ \FFFF } = n }
	= \bigcup_{ n_R \cdot n_S = n } 
		\set[\big]{ D \in \DB \with \card{ D }_{ \FFFF_R } = n_R \text{ and } 
		\card{ D }_{ \FFFF_S } = n_S } \in \DDB\text.\qedhere
\]
\end{proof}

\begin{rem}
	Note that we proceeded similar to \cref{lem:central}. Unfortunately, not
	every set $\TTTT \in \TT_Q$ can be decomposed into a product of measurable
	sets $\TTTT_R$ and $\TTTT_S$ as used above. The sets $\TTTT_R$ and $\TTTT_S$
	are the projections of $\TT_Q$ to $\TT_R$ and $\TT_S$. Recall that we 
	already encountered a similar situation in \cref{ex:nonmeas}. In general,
	such projections of measurable sets from products of standard Borel spaces
	to their factors yield the \emph{analytic} sets
	(cf.~\cite[Exercise~14.3]{Kechris1995}). A fundamental theorem in
	descriptive set theory states that every uncountable standard Borel space
	has analytic, non-(Borel-)measurable subsets
	\cite[Theorem~14.2]{Kechris1995}. Thus, with suitably chosen attribute
	domains, there are sets $\TTTT \in \TTT_Q$ such that, for instance, 
	the projection of $\TTTT$ to $\TT_R$ is not in $\TTT_R$. Therefore, the
	approach from \labelcref{eq:proddecomp} doesn't help us resolve the
	measurability of $Q$, even though it appeared promisingly similar to the
	arguments from the previous section. In particular, for the given query we
	cannot establish criterion \labelcref{itm:central2} from \cref{lem:central}.
\end{rem}

For $t \in \TT_Q$, we let $t_R \in \TT_R$ and $t_S \in \TT_S$ denote the 
projections of $t$ to its $R$- and $S$-part. Then for all $r > 0$, $\WWWW_r( t )
\coloneqq B_r(t_R) \times B_r(t_S) \subseteq \TT_Q$ is a \emph{measurable
rectangle} containing $t$ with $B_r( t_R ) \subseteq \TT_R$ and $B_r( t_S )
\subseteq \TT_S$. Note that $B_r( t_R )$ is a ball with respect to the (Polish)
metric on $\TT_R$, and $B_r( t_S )$ is a ball with respect to the (Polish)
metric on $\TT_S$.

\begin{rem}
	Let us briefly comment on the intuition of the setup. The sets $\WWWW_r( t
	)$ should be thought of as a small \emph{windows} that we can use to make
	our considerations local around $t = (t_R,t_S)$. We use these windows to
	show query measurability by exploiting the finiteness of our database
	instances: since every $D$ is finite, $D$ is $\epsilon$-coarse for $\epsilon
	> 0$ small enough (see \cref{ssec:coarse}). Then for small enough radius
	$r$, the balls $B_r( t_R )$ and $B_r( t_S )$ both contain at most one $R$-
	or $S$-tuple from $D$, respectively. Thus, the image of $D$ under the cross
	product query also contains at most one tuple in $\WWWW_r(t) = B_r( t_R )
	\times B_r( t_S )$. Then, as $\WWWW_r(t)$ has the appropriate shape,
	\cref{lem:crossprodrectangles} can be used again.
\end{rem}

Recall that for all $\epsilon > 0$, $\DB\rvert_\epsilon$ is the set of 
$\epsilon$-coarse instances in $\DB$. Let $\FFFF \in
\FFF_Q$ and $n \in \NN$. For all $\epsilon > 0$, we let
\[
	\DDDD_{ \FFFF, n, \epsilon } \coloneqq
	\set{ D \in \DB\rvert_\epsilon \with \card{ Q( D ) }_{ \FFFF } = n }
	\text.
\]
Then $\DDDD_{ \FFFF, n, \epsilon }$ is the $\epsilon$-coarse preimage of the
event $\set{ D_Q \in \DB_Q \with \card{ D }_{ \FFFF } = n }$ from $\DB_Q$. 

\begin{lem}\label[lem]{lem:FFFeps}
	For all $\FFFF \in \FFF_Q$, $n \in \NN$, $t \in \TT_{ Q }$ and all
	$\epsilon,r > 0$, with $r < \tfrac{\epsilon}3$ it holds that $\DDDD_{ \FFFF
	\cap R_Q(\WWWW_r( t )), n, \epsilon } \in \DDB$.
\end{lem}

\begin{proof}
	We prove the lemma as follows. Let
	\begin{equation}\label{eq:FFFeps}
		\FFF_{ \epsilon } \coloneqq
		\set[\big]{ \FFFF \in \FFF_Q \with
		\DDDD_{ \FFFF \cap R_Q( \WWWW_r( t ) ), n, \epsilon } \in \DDB \text{ for all }
			t \in \TT_Q\text,~n \in \NN\text{, and } r < \tfrac{\epsilon}3 }\text.
	\end{equation}
	Intuitively, this is the set of measurable sets $\FFFF$ with the property
	that all sufficiently fine \enquote{window approximations} of the
	$\epsilon$-coarse preimages of $\set{ D_Q \in \DB_Q \with \card{ D_Q
	}_{\FFFF} = n }$ are measurable in $\DDB$.

	It then suffices to show for all $\epsilon > 0$ that
	\begin{enumerate}
		\item \label[itm]{itm:goodsets1}
			for every set $\FFFF = R_Q( \TTTT )$ with $\TTTT = \TTTT_R \times 
			\TTTT_S \in \TTT_R \times \TTT_S \subseteq \TTT_Q$ it holds that 
			$\FFFF \in \FFF_{\epsilon}$, and
		\item \label[itm]{itm:goodsets2}
			the family $\FFF_{ \epsilon }$ is a $\sigma$-algebra on $\FF_Q$.
	\end{enumerate}
	From these two it follows that $\FFF_\epsilon = \FFF_Q$. That is, indeed 
	every measurable set $\FFFF \in \FFF_Q$ satisfies the property from 
	\labelcref{eq:FFFeps}. This line of argument is occasionally referred to as
	the \emph{good sets principle}~\cite[p.~5]{Ash1972}.

	We fix $\epsilon > 0$ arbitrary and show
	\labelcref{itm:goodsets1,itm:goodsets2}.
	\begin{enumerate}
		\item 
			Suppose $\FFFF = R_Q( \TTTT )$ for some $\TTTT \in \TTT_R \times 
			\TTT_S$, say $\TTTT = \TTTT_R \times \TTTT_S$. Let $t \in \TT_Q$,
			$n \in \NN$, and $r < \frac{\epsilon}3$ be arbitary. With $t_R \in
			\TT_R$ and $t_S \in \TT_S$, we denote the $R$- and the $S$-part of
			$t$, respectively. Then
			\begin{align*}
				\FFFF \cap R_Q\big( \WWWW_r( t ) \big)
				&= R_Q( \TTTT_R \times \TTTT_S ) \cap 
					R_Q( B_r( t_R ) \times B_r( t_S ) )\\
				&= R_Q\big( ( \TTTT_R \cap B_r( t_R ) ) \times 
					( \TTTT_S \cap B_r( t_S ) ) \big)
				\text.
			\end{align*}
			By \cref{lem:crossprodrectangles}, it follows that $\set{ D \in \DB 
			\with \card{ Q(D) }_{ \FFFF \cap R_Q( \WWWW_r(t) ) } = n } \in \DDB$.
			With \cref{lem:coarse} this entails that
			\[
				\DDDD_{ \FFFF \cap R_Q( \WWWW_r(t) ), n, \epsilon } = 
				\DB\rvert_{\epsilon} \cap \set{ D \in \DB \with 
				\card{ Q(D) }_{ \FFFF \cap R_Q( \WWWW_r( t ) ) } = n } \in \DDB
				\text.
			\]
			As $t$, $n$ and $r$ were arbitrary, it follows that $\FFFF \in
			\FFF_\epsilon$.
		\item 
			We show that $\FFF_{ \epsilon }$ is a $\sigma$-algebra on $\FF_Q$
			by showing $\FF_Q \in \FFF_{ \epsilon }$, and that it is closed under
			complements and countable intersections.

			First note that $\FF_Q \in \FFF_\epsilon$ follows from
			\cref{lem:crossprodrectangles} because $\FF_Q = R_Q( \TT_R \times
			\TT_S )$. Now let $t \in \TT_Q$, $n \in \NN$, and $r <
			\tfrac{\epsilon}3$ be arbitrary but fixed, and let $D \in
			\DB\rvert_\epsilon$.  \Cref{fig:FFFepsillustration} provides
			visualizations for the remaining cases.

\begin{figure}[H]
	\begin{subfigure}[b]{.5\linewidth}\centering%
		\begin{tikzpicture}
		\begin{scope}
			\clip (-.1, -.1) rectangle (3.1,3.1);
			\def\Sr{ (.5, .5) rectangle (2.5,2.5) }
			\def\origF{ plot coordinates {(-1.25,1.4) (-.2,1.2) (1.3,1.4)
				(1.5,-1.0) (-0.8,-1.3)} }
			\begin{scope}
				\fill[shade2] \Sr;
				\path[clip,smooth cycle] \origF;
				\fill[shade0] \Sr;
			\end{scope}
			\draw \Sr;
			\draw[smooth cycle] \origF;
		\end{scope}
		\node[anchor=west] at (-.5,.75) {$\FFFF$};
		\node[anchor=west] at (-.5,1.75) {$\FFFF^{\compl}$};
		\draw[stealth'-stealth',shade3] (0.5,2.7) to
			node[above] { $B_r(t_R)$ } (2.5, 2.7);
		\draw[stealth'-stealth',shade3] (2.7,2.5) to 
			node[right] { $B_r(t_S)$ } (2.7,0.5);
		\node[point,shade3,label={[shade3]above right:$t$}] at (1.5,1.5) {};
		\node[anchor=south west,rectangle] at (2.5,2.5) {$R_Q\big(\WWWW_r(t)\big)$};
		\end{tikzpicture}
		\caption{Depiction of $\FFFF^{\compl}\cap R_Q\big( \WWWW_r(t) \big)$.}
		\label{fig:windowcompl}
	\end{subfigure}%
	\begin{subfigure}[b]{.5\linewidth}\centering%
		\begin{tikzpicture}
		\begin{scope}
			\clip (-.1, -.1) rectangle (3.1,3.1);
			\def\Sr{ (.5, .5) rectangle (2.5,2.5) }
			\def\origFi{ plot coordinates {(-.4,2.1) (1.8,1.5)
				(1.1,-1.0) (-0.8,-1.3)} }
			\def\origFii{ plot coordinates {(-1.25,1.4) (-.2,1.4) 
			(2.2, .8) (2.7,-.5) (-0.8,-1.3)} }
			\begin{scope}
				\fill[shade0] \Sr;
				\path[clip,smooth cycle] \origFi;
				\path[clip,smooth cycle] \origFii;
				\fill[shade2] \Sr;
			\end{scope}
			\draw \Sr;
			\draw[smooth cycle] \origFi;
			\draw[smooth cycle] \origFii;
		\end{scope}
		\node[anchor=west] at (1.8,.2) {$\FFFF_2$};
		\node[anchor=west] at (-.7,1.8) {$\FFFF_1$};
		\draw[stealth'-stealth',shade3] (0.5,2.7) to
			node[above] { $B_r(t_R)$ } (2.5, 2.7);
		\draw[stealth'-stealth',shade3] (2.7,2.5) to 
			node[right] { $B_r(t_S)$ } (2.7,0.5);
		\node[point,shade3,label={[shade3]left:$t$}] at (1.5,1.5) {};
		\node[anchor=south west,rectangle] at (2.5,2.5) {$R_Q\big(\WWWW_r(t)\big)$};
		\end{tikzpicture}
		\caption{Depiction of $\FFFF_1 \cap \FFFF_2 \cap R_Q\big(\WWWW_r( t )\big)$.}
	\end{subfigure}
	\caption{Illustrations for the proof of \cref{lem:FFFeps}. The indications
	of $B_r(t_R)$ and $B_r(t_S)$ are only to convey the
	intuition behind the construction of
	$\WWWW_r(t)$.}\label{fig:FFFepsillustration}
\end{figure}

			\begin{enumerate}
				\item\label{itm:goodsetscompl} Let $\FFFF \in \FFF_\epsilon$. Then
					it holds that
					\begin{align*}
						&D \in \DDDD_{ \FFFF^{ \compl } \cap R_Q(\WWWW_r( t )), n, \epsilon }\\
						&\iff \card{ Q(D) }_{ \FFFF^{ \compl } \cap R_Q(\WWWW_r( t )) } = n\\
						&\iff \card{ Q(D) }_{ R_Q(\WWWW_r( t )) } - \card{ Q(D) }_{ \FFFF
						\cap R_Q(\WWWW_r( t )) } = n
					\end{align*}
				
					Recall that $D$ is $\epsilon$-coarse. In particular, as $r <
					\frac{\epsilon}3$, the ball $B_r( t_R )$ contains at most one
					$R$-tuple from $D$ and the ball $B_r( t_S )$ contains at most
					one $S$-tuple from $D$. Hence, $\WWWW_r( t ) = B_r( t_R )
					\times B_r( t_S )$ contains at most one $R_Q$-tuple from $Q(D)$.
					Thus, $D \in \DDDD_{ \FFFF^{ \compl } \cap R_Q(\WWWW_r( t )), n,
					\epsilon }$ if and only if 
						\[
							\card{ Q(D) }_{ R_Q(\WWWW_r( t )) } = n \text{ and }
							\card{ Q(D) }_{ \FFFF \cap R_Q(\WWWW_r( t )) } = 0\text.
						\]
					From \labelcref{itm:goodsets1} and the assumption $\FFFF \in
					\FFF_{\epsilon}$, it thus follows that $\FFFF^{ \compl } \in
					\FFF_\epsilon$.
				\item Now let $\big( \FFFF_i \big)_{ i = 1, 2, \dots }$ be a 
					sequence of sets $\FFFF_i \in \FFF_\epsilon$ and let $\FFFF
					= \bigcap_{ i = 1 }^{ \infty } \FFFF_i$. Then it holds that
					\begin{align*}
						&D \in \DDDD_{ \FFFF \cap R_Q(\WWWW_r( t )), n, \epsilon }\\
						&\iff \card{ Q(D) }_{ \FFFF \cap R_Q(\WWWW_r( t )) } = n\\
						&\iff \card{ Q(D) }_{ \bigcap_{ i = 1 }^{ \infty } 
						( \FFFF_i \cap R_Q(\WWWW_r( t )) ) } = n\text.
					\end{align*}
					As in the previous case, because $D$ is $\epsilon$-coarse, the
					set $\WWWW_r(t)$ contains at most one $R_Q$-tuple from $Q(D)$.
					Thus, $D \in \DDDD_{ \FFFF \cap R_Q(\WWWW_r( t )), n, \epsilon
					}$ if and only if
					\[
						\card{ Q(D) }_{ \FFFF_i \cap R_Q(\WWWW_r( t )) } = n
						\text{ for all } i = 1, 2, \dots\text.
					\]
					Thus, $\bigcap_{ i = 1 }^{ \infty } \FFFF_i \in \FFF_\epsilon$
					using the assumption.
			\end{enumerate}
			Together, $\FFF_\epsilon$ is indeed a $\sigma$-algebra on $\FF_Q$.
	\end{enumerate}
	From \labelcref{itm:goodsets1,itm:goodsets2} it follows that $\FFF_\epsilon 
	= \FFF_Q$ for all $\epsilon > 0$.
\end{proof}

\begin{lem}\label[lem]{lem:Fnepsilon}
	For all $\FFFF \in \FFF_Q$, all $n \in \NN$, and all $\epsilon > 0$, it 
	holds that $\DDDD_{ \FFFF, n, \epsilon } \in \DDB$.
\end{lem}

\begin{proof}
	Let $\FFFF \in \FFF_Q$, $n \in \NN$, and $\epsilon > 0$. It suffices to show
	that
	\[
		\DDDD_{ \FFFF, \geq n, \epsilon } \coloneqq
		\bigcup_{ m \geq n } \DDDD_{ \FFFF, m, \epsilon } 
		\in \DDB
	\]
	where we may assume $n > 0$. We show that $D \in \DDDD_{ \FFFF, \geq n, 
	\epsilon }$ is equivalent to $D$ satisfying the following condition.
	
	\begin{cond}\label[cond]{cond:FFFeps}
		The instance $D$ is $\epsilon$-coarse and for some $\ell \in \NN_+$ there
		are $k_1, \dots, k_\ell \in \NN_+$ with $k_1 + \dots + k_\ell \geq n$
		such that for all $r \in \big( 0, \frac\epsilon3 \big)$ there are
		$t_{1,r}^*, \dots, t_{\ell,r}^* \in \TT_R^* \times \TT_S^*$
		such that 
		\begin{enumerate}
			\item\label[itm]{itm:FFFepsi} 
				for all $i \neq j$ it holds that $d_R( t_{i,r,R}^*,
				t_{j,r,R}^* ) > \frac\epsilon3$ or $d_S( t_{i,r,S}^*,
				t_{j,r,S}^* ) > \frac\epsilon3$ and
			\item\label[itm]{itm:FFFepsii} 
				for all $i = 1, \dots, \ell$ it holds that $D \in \DDDD_{ \FFFF 
				\cap R_Q( \WWWW_r( t_{i,r}^* ) ) , k_i, \epsilon }$.
		\end{enumerate}
	\end{cond}
	
	We start with the easy direction ($\Leftarrow$).
	\begin{description}
		\item[$\Leftarrow$] Suppose $D$ satisfies \cref{cond:FFFeps}. Then $D \in 
			\DB\rvert_\epsilon$ and it remains to show $\card{ Q(D) }_{ \FFFF } 
			\geq n$. Note that it suffices to show that if $r$ is small enough in
			\cref{cond:FFFeps}, then the sets $\WWWW_r( t_{i,r}^* )$ are
			pairwise disjoint. In this case, the claim follows from
			\labelcref{itm:FFFepsii}.

			By \labelcref{itm:FFFepsi}, for all $i,j = 1, \dots, n$ with $i \neq 
			j$ it holds that at least one of $d_R( t_{i,R}^*, t_{j,R}^* )$ or
			$d_R( t_{i,S}^*, t_{j,S}^* )$ is larger than $\frac\epsilon3$. Thus,
			for $r < \frac\epsilon6$ it follows that at least one of 
			$d_R( t_{i,R}^*, t_{j,R}^* )$ or $d_S( t_{i,S}^*, t_{j,S}^* )$ is
			larger than $2r$. Therefore, $\WWWW_r( t_{i,r}^*) \cap \WWWW_r(
			t_{j,r}^* ) = \emptyset$.
		\item[$\Rightarrow$] Suppose that $D \in \DDDD_{ \FFFF, \geq n, \epsilon 
			}$. Then $D$ is $\epsilon$-coarse and $\card{ Q(D) }_{ \FFFF } \geq 
			n$. Thus, there exist pairwise distinct $t_1, \dots, t_\ell \in Q(D)$ 
			with $R_Q(t_i) \in \FFFF$ for all $i = 1, \dots, n$ such that
			\[
				\card{ Q(D) }_{ R_Q( t_1 ) } +
				\dots +
				\card{ Q(D) }_{ R_Q( t_\ell ) }
				\geq n\text.
			\]
			Let $r \in \big( 0, \frac\epsilon3 \big)$. Since $\TT_R^*$ is dense
			in $\TT_R$ and $\TT_S^*$ is dense in $\TT_S$, for all $i = 1,\dots,n$
			we can choose $t_{i,r}^* = ( t_{i,r,R}^*, t_{i,r,S}^* ) \in \TT_R^*
			\times \TT_S^*$ such that
			\[
				d_R( t_{i,r,R}^*, t_{i,R} ) < r
				\centertext{and}
				d_S( t_{i,r,S}^*, t_{i,S} ) < r
				\text.
			\]
			Thus, $t_i \in \WWWW_r( t_{i,r}^* )$. Now let $t \in Q(D)$ with $t
			\neq t_i$. Because $D$ is $\epsilon$-coarse, it holds that
			\[
				d_R( t_{i,R}, t_R ) > \epsilon 
				\centertext{or}
				d_S( t_{i,S}, t_S ) > \epsilon 
				\text.
			\]
			For the rest of the proof we assume $d_R( t_{i,R}, t_R ) > \epsilon$
			(the other case is completely symmetric). Recall that $r <
			\frac\epsilon3$. Then it follows that
			\[
				d_R( t_{i,r,R}^*, t_R )
				\geq d_R( t_{i,R}, t_R ) - d_R( t_{i,R}, t_{i,r,R}^* )
				> \epsilon - r > r\text.
			\]
			Thus, $t \notin \WWWW_r( t_{i,r}^* )$, so every $\WWWW_r( t_{i,r}^* )$ 
			contains no tuple from $Q(D)$ other than $t_i$.	Thus, $\card{ Q(D) }_{ 
			R_Q(t_i) } = \card{ Q(D) }_{ \FFFF \cap R_Q(\WWWW_r( t_{i,r}^* )) }$,
			and \labelcref{itm:FFFepsii} follows.

			Also for all $j \neq i$, we have that
			\[
				d_R( t_{i,r,R}^*, t_{j,r,R}^* ) 
				\geq d_R( t_{i,R}, t_{j,R} ) 
					- d_R( t_{i,R}, t_{i,r,R}^* ) 
					- d_R( t_{j,R}, t_{j,r,R}^* )
				> \epsilon - 2r > \tfrac\epsilon3,
			\]
			establishing \labelcref{itm:FFFepsi}.
	\end{description}
	The equivalence still holds, if $r$ is additionally required to
	be rational. That is,
	\[
		\DDDD_{ \FFFF, \geq n, \epsilon } =
		\bigcup_{ \ell } \bigcup_{ k_1, \dots, k_\ell }
		\bigcap_{ r } \bigcup_{ t_{1,r}^*, \dots, t_{\ell,r}^* }
		\DDDD_{ \FFFF\cap R_Q(\WWWW_r( t_{i,r}^* )), k_i, \epsilon } \in \DDB
	\]
	using \cref{lem:FFFeps} (with the indices ranging as in our equivalence, 
	and numbers being restricted to rationals).
\end{proof}

Finally, the measurability of $Q = R \times S$ is a direct consequence of
\cref{lem:Fnepsilon,cor:coarsepreimage}.

\begin{lem}\label[lem]{lem:crossprod}
	The query $Q = R \times S$ is measurable.
	\qed
\end{lem}

\bigskip

Altogether,
\cref{lem:base,lem:basicbagops,lem:dedupe,lem:selection,lem:projection,lem:crossprod}
now prove the main result---the measurability of $\mathsf{BALG}^1$-queries---of
this section (\cref{thm:balg}).

As a consequence, we also obtain the measurability of all kinds of derived 
operators, including the typical join operators. Note that it also follows from 
\cref{thm:balg}, that we can use finite Boolean combinations of predicates for
selection queries.

\section{Aggregation}\label{sec:agg}

There are practically relevant queries that are not already covered by our
treatment in the previous section. For example, in our running example of
temperature recordings, we might be interested in returning the \emph{average}
(or \emph{minimum} or \emph{maximum}) temperature per room, taken over all
temperature records for this particular room.

In this section, we formalize aggregate operators and aggregate queries,
possibly with grouping in the standard PDB framework in a possible worlds 
semantics style (cf. \cref{ssec:pws}). In particular, we show that these
queries are measurable in the standard PDB framework.

\begin{rem}
	Often, when a separate treatment of aggregate queries over purely algebraic
	ones is motivated, it is mentioned that the correspondence of relational
	algebra and relational calculus limits expressive power to that of 
	first-order logic. However, bag query languages based on relational algebra
	\emph{do} allow expressing various kinds of aggregation based on exploiting
	the presence of multiplicities. For example, counting in a unary fashion is
	possible in $\mathsf{BALG^1}$~\cite{GrumbachMilo1996}. Here, we follow a 
	more general approach in allowing basically any measurable function over 
	finite bags to be used for aggregation. This goes beyond the integer
	aggregation of $\mathsf{BALG}$~\cite{Grumbach+1996}.
\end{rem}

Let $(\AA, \AAA)$ and $(\BB,\BBB)$ be standard Borel spaces. An \emph{aggregate
operator} (or \emph{aggregator}) from $\AA$ to $\BB$ is a function $\Phi \from
\powerbag_{ \fin }( \AA ) \to \BB$, that is, a function mapping (finite) bags
of elements of $\AA$ to elements of $\BB$. If $\TT_R = \AA$, then every such
aggregator $\Phi$ induces an \emph{aggregation query} $Q = \bagagg{ \Phi }{ R
}$ with $\TT_Q = \BB$ as follows:
\[
	\card{ Q(D) }_{ R_Q( t ) } \coloneqq 
	\begin{cases}
		1 & \text{if } t = \Phi\big( R( D ) \big)\text{, and}\\
		0 & \text{otherwise,}
	\end{cases}
\]
for all $D \in \DB$.\footnote{Formally, $R(D)$ has been defined as the bag of
\emph{$R$-facts} in $D$ wheras $\Phi$ should take bags of \emph{$R$-tuples}.
This small type mismatch is of no significance whatsoever.}
In essence, its output on a single instance $D$ is the instance containing only
the single \enquote{tuple} corresponding to the value of the aggregation over
the relation $R$ in $D$. Examples of common attribute operators $\Phi$ are
shown in \cref{tab:commonagg}.

\begin{prop}\label[prop]{prop:aggregation}
	If $\Phi$ is a $\big( \Count( \TT_R ), \TTT_Q \big)$-measurable aggregator,
	then $\bagagg{ \Phi }{ R }$ is a measurable query.
\end{prop}

\begin{proof}
	Let $\FFFF = R_Q( \TTTT )$ where $\TTTT \in \TTT_Q$. Without loss of
	generality, we assume that $R$ is the only relation in $\tau$. Then for all
	$D \in \DB$, it holds that
	\[
		\card{ Q(D) }_{ \FFFF } = 1
		\iff
		R( D ) \in \Phi^{-1}( \TTTT )
		\iff
		D \in R\big( \Phi^{-1}( \TTTT ) \big)
	\]
	(and $\card{ Q(D) }_{ \FFFF } = 0$ otherwise). Since $\Phi$ is $\big( 
	\Count( \TT_R ), \TTT_Q \big)$-measurable, the claim follows.
\end{proof}

\begin{table}[H]
	\centering
	\caption{Common aggregate operators. (We assume that $\MIN$, $\MAX$ and 
	$\AVG$ have a suitable definition on empty bags, so that their semantics are
	well-defined.)}\label{tab:commonagg}
	\begin{tabular}{ll}
		\toprule
		\textbf{Name}	& \textbf{Definition}\\ \midrule
		Count				& $\CNT( \bag{ a_1, \dots, a_m } ) 
								\coloneqq m$\\
		Distinct Count	& $\CNTd( \bag{ a_1, \dots, a_m } )
								\coloneqq \card{ \set{ a_1, \dots, a_m } }$\\
		Sum				& $\SUM( \bag{ a_1, \dots, a_m } ) 
								\coloneqq a_1 + \dotsc + a_m$\\
		Minimum			& $\MIN( \bag{ a_1, \dots, a_m } )
								\coloneqq \min\set{ a_1, \dots, a_m }$\\
		Maximum			& $\MAX( \bag{ a_1, \dots, a_m } )
								\coloneqq \max\set{ a_1, \dots, a_m }$\\
		Average			& $\AVG( \bag{ a_1, \dots, a_m } ) 
								\coloneqq \frac{ a_1 + \dots + a_m }{ m }$\\
		\bottomrule
	\end{tabular}
	\label{tab:common_aggregators}
\end{table}

For $m \in \NN$, we call a function $\phi \from \TT_R^m \to \TT_Q$
\emph{symmetric} if $\phi( t ) = \phi( t' )$ for all $t \in \TT_R^m$ and all
permutations $t'$ of $t$.

\begin{lem}
	For all $m \in \NN$, let $\phi_m \from \TT_R^m \to \TT_Q$ be a
	symmetric, measurable function. Then the aggregator $\Phi \from \powerbag_{
	\fin }( \TTT_R ) \to \TT_Q$ with
	\[
		\Phi\big( \bag{ a_1, \dots, a_m } \big) 
		\coloneqq \phi_m( a_1, \dots, a_m )
	\]
	is $\big( \Count( \TT_R ), \TTT_Q )$-measurable.
\end{lem}

\begin{proof}
	It suffices to show that the restriction $\Phi_m$ of $\Phi$ to 
	$\powerbag_m( \TT_R )$, the bags of cardinality $m$ over $\TT_R$, is
	measurable for all $m \in \NN$. Note that $\powerbag_0( \TT_R )$ only
	contains the empty bag, and $\TT_R^0$ only contains the empty tuple. That
	is, the statement is trivial for $m=0$. Thus, let $m \in \NN_+$. Since
	$\phi_m$ is $( \TTT_R^{\otimes m}, \TTT_Q )$-measurable and symmetric, for
	all $\TTTT \in \TTT_Q$ it holds that $\phi_m^{-1}( \TTTT)$ is a symmetric
	set in $\TTT_R^{\otimes m} \subseteq \bigoplus_{ m = 0 }^{ \infty } \TTT_R^{
	\otimes m }$. From \cref{rem:symborel,fac:symcounting} it follows that
	$\Phi_m^{-1}( \TTTT ) \in \Count( \TT_R )$.
\end{proof}

\begin{exa}\label[exa]{ex:common_aggregators}
	All the aggregators in \cref{tab:common_aggregators} yield measurable
	aggregation queries. The associated functions $\phi_m$ are all continuous
	under suitable choices of attribute domains. That is, for example, if $+$ in 
	the definition of $\SUM$ is the addition of real numbers.
\end{exa}

What we have introduced so far is only sufficient to express the aggregation
over all tuples of a relation at once. Usually, we want to perform aggregation
separately for parts of the data, as in our motivating example of returning the
average temperature per room. For this, we need to \emph{group} tuples before
aggregating values. Suppose we want to group a relation $R$ by attributes $A_1,
\dots, A_k$ and perform the aggregation over attribute $A$, separately for
every occurring value of the attributes $A_1, \dots, A_k$. Without loss of 
generality, we assume that $\sort( R ) = ( A_1, \dots, A_k, A )$. Then what we
described is an \emph{group-by aggregate query} $\bagagg{ A_1, \dots, A_k,
\Phi( A ) }{ R }$ and is defined by
\begin{equation}\label{eq:groupagg}
	\card{ Q(D) }_{ R_Q( a_1, \dots, a_k,b ) } 
	\coloneqq
	\begin{cases}
		1 & \text{if } b = \Phi( \AAAA_{a_1,\dots,a_k}(D) )\text{ and}\\
		0 & \text{otherwise.}
	\end{cases}
\end{equation}
for all $D \in \DB$ where $\AAAA(D)$ is the bag of values $a$ such that
$R(a_1,\dots,a_k,a) \in D$ with
\[
	\card{ \AAAA_{a_1,\dots,a_k}(D) }_a = \card{ D }_{ R( a_1,\dots,a_k,a ) }\text.
\]
Essentially, $\AAAA_{a_1,\dots,a_k}(D)$ is obtained by selecting those tuples
where the first $k$ attributes have values $a_1,\dots,a_k$, and then projecting
to the last attribute. Both kinds of aggregate queries (without, and with
grouping) are shown in \cref{tab:aggregation}.

\begin{table}[H]
	\centering
	\caption{Aggregate Queries.}\label{tab:aggregation}
	\begin{tabular}{ ll }
		\toprule%
		\textbf{Query}
		& \textbf{Semantics}	(for all $D \in \DB$, $t, (a_1,\dots,a_k,b) \in 
		\TT_Q$)									
		\\%
		\midrule%
		$Q = \bagagg{ \Phi }{ R }$				
		& $\card{ Q(D) }_{ R_Q(t) } \coloneqq 1$ if $t = \Phi\big( R(D) \big)$,
		and $0$ otherwise
		\\%
		$Q = \bagagg{(A_1,\dots,A_k,\Phi(A))}{R}$
		& $\card{ Q(D) }_{ R_Q(a_1,\dots,a_k,b) } \coloneqq 1 \text{ if } b =
		\Phi( \AAAA_{a_1,\dots,a_k}(D) ) \text{ and } 0 \text{ otherwise}$
		\\%
		\bottomrule%
	\end{tabular}
\end{table}

\begin{thm}\label[thm]{thm:groupagg}
	If $\Phi \from \powerbag_{ \fin }( \AA ) \to \BB$ is a $( \Count( \AA ), 
	\BBB )$-measurable aggregator, then the query $\bagagg{ A_1, \dots, A_k,
	\Phi( A ) }{ R }$ is a measurable query.
\end{thm}
\def\grp{\mathsf{grp}}

\begin{proof}
	Let $\TT_{\grp} = \prod_{ i = 1 }^{ k } \AA_i$. In the following, we fix a
	compatible Polish metric $d_{\grp}$ on $\TT_{\grp}$, and a countable dense
	set $\TT_{\grp}^*$ in $\TT_{\grp}$. For all $t \in \TT_{\grp}$ and all $r >
	0$ define 
	\[
		Q_{t,r} =
			\bagproj[\big]{ A_1, \dots, A_k }{%
				\bagsel{ (A_1, \dots, A_k) \in B_r( t ) }{ R }
			}
			\times
			\bagagg[\big]{\Phi}{
				\bagproj[\big]{ A }{
					\bagsel{ (A_1, \dots, A_k) \in B_r( t ) }{
						R
					}
				}
			}
		\text.
	\]
	Note that $Q_{t,r}$ is measurable by \cref{thm:balg} for all particular
	choices of $t$ and $r$. Let $\widetilde Q = \bagproj{ A_1,\dots,A_k }{R}$
	and
	\[
		\widetilde{\DB}\rvert_{\epsilon} 
		\coloneqq 
		\set{ D \in \DB \with \widetilde Q( D ) \text{ is $\epsilon$-coarse}}
	\]
	Then $\widetilde{ \DB }\rvert_\epsilon \in \DDB$, as
	$\widetilde{\DB}\rvert_\epsilon = \widetilde{Q}^{-1}\big(
	\DB'\rvert_{\epsilon}\big)$ using \cref{lem:projection,lem:coarse} where
	$\DB'$ is the output instance space of $\widetilde Q$. Intuitively,
	$\widetilde{\DB}\rvert_\epsilon$ are the instances that are
	$\epsilon$-coarse in the $(A_1,\dots,A_k)$ attributes of $R$. Similar to
	\cref{cor:coarsepreimage} it suffices to show that $\set{ D \in \widetilde{
	\DB }\rvert_\epsilon \with \card{Q(D)}_{\FFFF} \geq n }$ is measurable for
	all positive $\epsilon$, $\FFFF \in \FFF_Q$ and $n \in \NN_+$.

	We show that for all $D \in \widetilde{\DB}\rvert_\epsilon$, all $\FFFF \in
	\FFF_Q$, and all $n \in \NN_+$ it holds that $\card{ Q(D) }_{ \FFFF } \geq n$
	is equivalent to $D$ satisfying the following condition.

	\begin{cond}\label[cond]{cond:groupaggcond}
		For all positive $r < \tfrac\epsilon3$ there exist $t_{1,r}^*, \dots,
		t_{n,r}^* \in \TT_{\grp}^*$ with $d_{\grp}( t_{i,r}^*, t_{j,r}^* ) >
		\tfrac\epsilon3$ such that $\card{ Q_{t_{i,r}^*,r}(D) }_{\FFFF} \geq 1$
		for all $i =1, \dots, n$.
	\end{cond}

	\begin{description}
		\item[$\Rightarrow$] Let $D \in \widetilde{\DB}\rvert_\epsilon$ such that 
			$\card{ Q(D) }_{ \FFFF } \geq n$. Then there exist $R_Q( t_1, b_1 ),
			\dots, R_Q( t_n, b_n ) \in \FFFF$ with $\card{ Q(D) }_{ R_Q( t_i, b_i
			) } = 1$ for all $i = 1, \dots, n$ such that $d_{\grp}( t_i,t_j ) >
			\epsilon$ for all $i\neq j$. Since $\TT_{\grp}^*$ is dense in
			$\TT_{\grp}$, for all positive $r$ (in particular $r <
			\tfrac\epsilon3$) there exist $t_{i,r}^*$ with $d_{\grp}( t_i,
			t_{i,r}^* ) < r$ for all $i = 1, \dots, n$. Since $d_{\grp}( t_i, t_j)
			> \epsilon$, it follows that $d_{\grp}( t_{i,r}^*, t_{j,r}^* ) >
			\tfrac\epsilon3$.  In particular, every $B_r(t_{i,r}^*)$ contains no
			tuple among $t_1, \dots, t_n$ other than $t_i$. Thus,
			\[
				\card{ Q_{t_{i,r}^*,r}(D) }_{ R_Q( t_i,b_i ) } = 1 \text,
			\]
			i.\,e.\ $\card{ Q_{ t_{i,r}^*, r }(D) }_{ \FFFF } \geq 1$ for all $i =
			1, \dots, n$.
		\item[$\Leftarrow$] Suppose \cref{cond:groupaggcond} holds. Since $D$ is
			finite, the tuples $t_{i,r}^*$ in \cref{cond:groupaggcond} converge to 
			tuples $t_i$ with
			\begin{equation}\label{eq:condrtl}
				\card{ Q_{t_{i,r}^*,r}(D) }_{ R_Q( t_i,b_i ) } \geq 1
			\end{equation}
			for some $b_i \in \AA$ where $R_Q(t_i,b_i) \in \FFFF$. Because $D \in 
			\widetilde{\DB}\rvert_\epsilon$, and $d_{\grp}( t_{i,r}^*, t_{j,r}^* )
			> \tfrac\epsilon3$, the tuples $t_1, \dots, t_n$ are pairwise
			distinct. Thus, \labelcref{eq:condrtl} implies $\card{Q(D)}_{ \FFFF }
			\geq \sum_{ i = 1 }^{ n } \card{ Q_{t_{i,r}^*,r}(D) }_{ \FFFF } \geq
			n$.
	\end{description}

	\bigskip

	The equivalence still holds, when $r$ is additionally required to be 
	rational. Thus,
	\[
		\set{ 
			D \in \widetilde{\DB}\rvert_\epsilon \with 
			\card{ Q(D) }_{ \FFFF } \geq n 
		} 
		= 
		\bigcap_r \bigcup_{t_{1,r}^*,\dots,t_{n,r}^*} \bigcap_{i} 
		\set{ 
			D \in \widetilde{\DB}\rvert_\epsilon \with 
			\card{ Q_{t_{i,r}^*,r}(D) }_{ \FFFF } \geq 1
		}\text.
	\]
	with the indices ranging as in \cref{cond:groupaggcond} (and $r$ rational).
\end{proof}

By the observation of \cref{ex:common_aggregators}, \cref{thm:groupagg} applies 
to the operators from \cref{tab:common_aggregators}.

\begin{cor}
	The query $\bagagg{ A_1,\dots,A_k,\Phi( A ) }{ R }$ is measurable for all 
	the aggregate operators $\Phi \in \set{ \CNT, \CNTd, \SUM, \MIN, \MAX, \AVG
	}$.
	\qed
\end{cor}

\section{Datalog}\label{sec:dl}

The measurability results of the previous sections also allow us to say
something about fixpoint queries. The key observation is the following lemma,
which follows from \cref{fac:topomeas}\labelcref{itm:topomeas2}.

\begin{lem}
	Let $( Q_i )_{ i \in \NN }$ be a family of measurable queries such that
	$Q( D ) \coloneqq \bigcup_{ i = 0 }^{ \infty } Q_i( D )$ is finite for all
	$D \in \DD$. Then $Q$ is a measurable query.
\end{lem}

\begin{proof}
	For all $n \in \NN$, let $Q^{ ( n ) } \coloneqq \bigcup_{ i = 0 }^{ n } 
	Q_i$. As a finite (maximum-)union of measurable queries, $Q^{ ( n ) }$ is
	measurable. As $Q(D)$ is finite for all $D$, it holds that $Q(D) = \lim_{n
	\to \infty} Q^{(n)}( D )$, and $Q$ is the pointwise limit of the functions
	$Q^{ (n) }$. Thus, $Q = \lim_{ n \to \infty } Q^{ ( n ) }$ is measurable as
	well.
\end{proof}

We omit the definition of Datalog and related query languages. For simplicity,
we consider set PDBs, and Datalog with sets semantics. Recall that if $Q$ is a
Datalog query, then $Q$ can be written as a countable union of conjunctive
queries~\cite{Abiteboul+1995}.\footnote{We note that a similar statement can be
  made for the bounded fixpoints semantics over bags featured
  in~\cite{ColbyLibkin1997}.} Thus, combining our measurability results with the
above lemma, we obtain the following.

\begin{cor}\label[lem]{lem:countable-union}
	Every Datalog query is measurable.  \qed
\end{cor}

In fact, our argumentation can be applied to all types of queries with
operators that are based on countable iterative (or inductive, inflationary, or
fixed-point) processes. All we need is that the iterative mechanism forms a 
converging sequence of measurable queries. For partial Datalog / fixed-point
logic, we cannot directly use \cref{lem:countable-union}, but a slightly more
complicated argument still based on countable limits works there as well.

\section{Beyond Possible Worlds Semantics}\label{sec:beyond}

In \cref{ssec:pws}, we introduced the notion of queries or views on 
probabilistic databases solely based on the existing notion of queries and
views for traditional databases, which we referred to as the \emph{possible
worlds semantics} of queries or views. As the title of this section suggests,
we explicitly leave this setup.\footnote{Note that we \emph{do not} abandon our
definition of PDBs as probability spaces over possible worlds. It is only that
we broaden the notion of views to also incorporate mappings of PDBs that are
not defined on a \enquote{per possible world} basis.} Before, we have 
introduced views as functions mapping database instances to database instances
and adopted a semantics based on possible worlds. Now, we want to discuss
\emph{PDB views} as functions that map probabilistic databases to probabilistic
databases, for which no such semantics (to be precise, a definition as in in 
the shape of \cref{ssec:pws}) exists. Such \enquote{views} naturally
arise in a variety of computational problems in probabilistic databases. For
example, consider the following problems or \enquote{queries}:
\begin{itemize}
	\item \emph{probabilistic threshold queries} that intuitively return
		a deterministic table containing only those facts which have a 
		marginal probability over some specified threshold~\cite{Qi+2010};
	\item \emph{probabilistic top-$k$ queries} that intuitively return a
		deterministic table containing the $k$ most probable facts~\cite{Re+2007}; 
	\item \emph{probabilistic skyline queries}~\cite{Pei+2007} that consider
		how different instances compare to each other with respect to
		some notion of \emph{dominance}; and
	\item \emph{conditioning}~\cite{KochOlteanu2008} the
		probabilistic database to some event.
\end{itemize}
Note that the way we informally explained the first two queries above is only
sensible if the space of facts is discrete. In a continuous setting, we
interpret these queries with respect to a suitable countable partition of the
fact space into measurable sets.

More such \enquote{queries} as the above can be found
in~\cite{Aggarwal2009,Wang+2013}.  These queries (or views) still take as input a
PDB and produce some output, but differ from the ones we have seen so far in
that they have no reasonable semantics on \emph{single} instances (i.\,e.\ per
possible world). This can be, for example, because they explicitly take
probabilities into account. The goal of this section is to interpret other
kinds of problems on PDBs abstractly as functions on probability spaces in
order fit them into a unified framework.  Developing such an understanding has
already been motivated in~\cite{WandersVanKeulen2015} as yielding potential
insight into common properties of the corresponding problems.

We now present our formal classification of views that are directly defined on
the probability space level of a PDB (as opposed to the instance level as in
\cref{ssec:pws}). Let $\PDBs_{\tau}$ denote the class of probabilistic
databases of schema $\tau$. Note that all PDBs in $\PDBs_{\tau}$ have the same
instance measurable space $(\DB,\DDB)$. Queries and, more generally, views of
input schema $\tau$ and output schema $\tau'$ are now mappings
$V\colon\PDBs_{\tau}\to\PDBs_{\tau'}$.

\pagebreak 
\newcommand*{\VT}[1]{\bm{\mathsf{V}}\sb{\text{\normalfont\uppercase\expandafter{\romannumeral #1}}}}

\begin{defi}[View Types]\leavevmode\label[defi]{def:vt}
	\begin{enumerate}
		\item Every view is \emph{type 1}.
		\item A view $V$ is \emph{type 2}, or \emph{pointwise local}, if (and 
			only if) for every fixed input PDB $\D = ( \DB, \DDB, P )$ there
			exists a measurable function $q_{ \D } \from \DB \to \DB_V$ such that
			$P_V = P \after q_{ \D }^{ -1 }$.\label[itm]{itm:vt2}
		\item A view $V$ is \emph{type 3}, or \emph{uniformly local}, if there
			exists a measurable function $q \from \DB \to \DB_V$ such that for
			every input PDB $\D = ( \DB, \DDB, P )$, it holds that $P_V = P
			\after q^{-1}$.\label[itm]{itm:vt3}
		\item A view $V$ is \emph{type 4}, or \emph{pointwise}, if (and only if)
			there exists measurable function $q \from \FFFF_q \to \FF_Q$ such
			that $V$ is composed of functions of the shape of $Q$ as in 
			\cref{thm:mapping} (where $Q$ depends on $q$).
	\end{enumerate}
	We let $\VT1$, $\VT2$, $\VT3$ and $\VT4$ denote the classes of views of type
	1 throughout 4.
\end{defi}

\begin{rem}
	Let us shed some more light on these classes and their names:
	\begin{enumerate}
		\item Class $\VT1$ does not require further explanation, as it contains
			\emph{every} view.

		\item The class $\VT2$ is described via functions that may depend on the
			measurable structure of the input PDB. Specifically, these views are 
			functions that directly transform input to output PDBs. We dub this
			\enquote{pointwise local}, because this function is applied per
			instance (hence, local) but the function itself depends on the
			concrete PDB (hence, is only pointwise local with respect to
			probability spaces). In general, type 2 views may take the probability
			space level into account. For example, probabilistic threshold or
			probabilistic top-$k$ queries can be viewed as views of this class:
			they transform any input PDB to a single database instance (that is, a
			PDB with only one possible world of probability $1$) containing the
			respective output tuples along with their probability in a separate
			attribute.  Another example is conditioning a PDB, as the probability
			measure of a conditioned PDB involves a normalization term that
			depends on the probability mass of an event.

		\item The class $\VT3$ captures the lifting of typical database queries
			to PDBs under the possible worlds semantics. Hence, there is a
			\emph{single} (measurable) function that is applied \enquote{locally}
			on every database instance. The term \enquote{uniformly} expresses 
			that it does not depend on the concrete PDB. In general, type 3 views
			only take the instance level into account. All the views that we
			investigated in \cref{sec:ra,sec:agg,sec:dl} fall into this category.

		\item The class $\VT4$ is the class of views corresponding to the
			measurability criterion of the Mapping Theorem. That is, there is a
			(measurable) function that is applied \enquote{pointwise}%
			\footnote{Note that the \enquote{pointwise} in the term
				\enquote{pointwise local} from $\VT2$ refers to probability spaces
				as \enquote{points}, whereas the term \enquote{pointwise} alone, as
				here in the definition of $\VT4$ refers to facts as
			\enquote{points}}
			on every fact. This transformation naturally lifts to instances and,
			thus, to PDBs. In general, type 4 views only take the fact level into
			account.  We have seen an example of such a view in
			\cref{exa:running3}.
	\end{enumerate}
\end{rem}

\begin{exa}\label[exa]{exa:running4}
	Recall our running example of temperature measurements. We use this to 
	introduce an example of a view that performs \enquote{out-of-world
	aggregation}~\cite{WandersVanKeulen2015} (and that, in particular, is of
	type $2$, but not of type $3$). The relation $\REL{TempRec}$ stores triples
	of room numbers ($\ATT{RoomNo}$), recording dates ($\ATT{Date}$) and
	recorded temperatures ($\ATT{Temp}$). Assume that the pair
	$(\ATT{RoomNo},\ATT{Date})$ acts as a key, so that with probability $1$
	there is at most one temperature recording per pair. Moreover, assume that
	the PDB is modelled with independent tuples $(r,d,\theta)$ where $\theta$ is
	Normally distributed per pair $(r,d)$, but such that the existence of an
	record belonging to $(r,d)$ is subject to uncertainty. That is, we have a mix
	of attribute- and tuple-level uncertainty. A possible representation of such
	a PDB is shown in \cref{fig:temprecexp}. Therein, we have two possible
	tuples (one recording for room \STR{4108} and one recording for room
	\STR{4108a}), but the value of the temperature recording is specified as a
	Normally distributed random variable, parameterized with its mean and
	variance. We assume the existence of tuples with room \STR{4108} or
	\STR{4108a} to be independent from one another, and that these events carry
	probability $0.4$ and $0.8$, respectively.

	\begin{figure}[H]
		\begin{tcbraster}[raster columns=2, raster force size=false, raster
			column skip=1cm, raster halign=center, nobeforeafter]
		\begin{tcolorbox}[hbox,fancy table, tabular={c c c | c},
			title={\strut$\REL{TempRec}$},before upper app={\rowcolor{llgray}}]
			\ATT{RoomNo}	& \ATT{Date}				& \ATT{Temp [\textdegree{}C]} & \ATT{Prob}\\\hline\hline
			\STR{4108}		& \STR{2021-07-12}		& $\N(20.5,0.1)$ & $0.4$\\\hline
			\STR{4108a}		& \STR{2021-07-12}		& $\N(21.0,0.1)$ & $0.8$\\
		\end{tcolorbox}
		\begin{tcolorbox}[hbox,fancy table, tabular={c c c | c},
			title={\strut$\REL{View}$},before upper app={\rowcolor{llgray}}]
			\ATT{RoomNo}	& \ATT{Date}				& \ATT{Temp [\textdegree{}C]} & \ATT{Prob}\\\hline\hline
			\STR{4108}		& \STR{2021-07-12}		& \STR{20.5}								& $0.4$\\\hline
			\STR{4108a}		& \STR{2021-07-12}		& \STR{21.0}								& $0.8$
		\end{tcolorbox}
		\end{tcbraster}
		\caption{Representation of a PDB along with the representation of a view
		that goes beyond the instance level.}\label{fig:temprecexp}
	\end{figure}

	A possible view could ask to return, per room and date, the \emph{expected}
	temperature, under preserving the tuple-level uncertainty. The output of
	this view is shown on the right-hand side of \cref{fig:temprec}. (It is
	here trivially to obtain, because the expected value was already part of
	the parametrization. The idea however also applies to more complicated
	setups.) The possible worlds of the input PDB $\D = (\DB, \DDB, P)$ are
	partitioned into four cases, depending on the presence of a tuple with room
	\STR{4108} and \STR{4108a}, respectively. Our function $q_{\D}$ (which
	depends on $\D$) maps instances $D \in \DB$ as follows: the tuple
	$(\STR{4108},\STR{2021-07-12},\STR{20.5})$ is present in the view result if and
	only if $D$ contains a tuple with room \STR{4108}, and the tuple
	$(\STR{4108},\STR{2021-07-12},\STR{21.0})$ is present in the view result if and
	only if $D$ contains a tuple with room \STR{4108a}. Then the output PDB
	has probability measure $P \after q_{\D}^{-1}$ and, in particular, four
	possible worlds. This is a so-called \emph{tuple-independent} PDB with 
	given marginal probabilities ($0.4$ and $0.8$), hence it can be represented as
	shown on the right-hand side of \cref{fig:temprec}.\footnote{The input PDB
	in this case is not tuple-independent, as any two facts with the same room
	number are mutually exclusive. Instead, this is a so-called (uncountable)
	\emph{block-independent disjoint} PDB with two independent blocks, both
	representing a single random tuple, such that each block specifies the
	(probability distribution over) possible manifestations for this tuple.}
\end{exa}

Let us come back to the relationships between the classes of views we defined
in \cref{def:vt}. Clearly, $\VT1 \supseteq \VT2 \supseteq \VT3 \supseteq
\VT4$. We already informally argued in \cref{sec:mapping} that queries that
depend on multiple tuples per instance are not captured by $\VT3$, i.\,e.\ $\VT3
\not\subseteq \VT4$. We provide two examples to expose that also the remaining
inclusions are strict.

\begin{prop}\label[prop]{pro:1not2}
	There exists a view that is not type 2, i.\,e.\ $\VT2 \subsetneq \VT1$.
\end{prop}

We demonstrate this using a view that \emph{conditions} its input PDB on an 
event.

\begin{proof}
	Let $\D$ be a PDB with three possible worlds $D_1$, $D_2$, and $D_3$, such
	that 
	\[
		P( \set{D_1} ) = \tfrac16\text,\quad
		P( \set{D_2} ) = \tfrac12\text{, and}\quad
		P( \set{D_3} ) = \tfrac13
		\text.
	\]
	Consider the view $V$ that conditions a PDB on the event $\set{ D_1, D_2 }$.
	Note that
	\[
		P_{ V(\D) }( \set{ D_1 } ) = \tfrac16 / \big( \tfrac16 + \tfrac12 \big) 
		= \tfrac14
		\centertext{and}
		P_{ V(\D) }( \set{ D_2 } ) = \tfrac12 / \big( \tfrac16 + \tfrac12 \big) 
		= \tfrac34
		\text.
	\]
	Yet, there is no event $\DDDD$ in $\D$ with the property that $P( \DDDD ) =
	\tfrac14$. Thus, a function $q_{ \D }$ as required in 
	\cref{def:vt}\labelcref{itm:vt2} does not exist and, hence, $V$ is not type 
	2.
\end{proof}

\begin{prop}\label[prop]{pro:2not3}
	There exists a type 2 view that is not type 3, so $\VT3 \subsetneq \VT2$.
\end{prop}

We have already discussed such a view in \cref{exa:running4}, but have not
actually shown that it is not of type $3$. We demonstrate the propoosition in
considering another example, namely that of a probabilistic threshold query.

\begin{proof}
	Let $\alpha \in (0, 1]$. We consider a function $V_{\alpha}$ that maps an
	input PDB $\D = ( \DB, \DDB, P )$ to an output PDB $V_{\alpha}( \D )$ such
	that 
	\[
		D_{\D,\alpha} \coloneqq 
		\set[\Big]{ 
			f \in \FF \with 
			P\big( \set{ D \in \DB \with \card{ D }_f > 0 }\big) \geq \alpha
		}
	\]
	has probability $1$ in $V_{\alpha}( \D )$. Note that $D_{\D,\alpha}$ is 
	finite for all $\alpha\in( 0, 1]$, so $V_{\alpha}$ is well-defined. The view
	$V_{\alpha}$ is a probabilistic threshold query with threshold
	$\alpha$. Consider the function $q_{ \D }$ with $D \mapsto D_{\D,\alpha}$
	for all $D \in \DB$. Then $q_{ \D }$ is measurable and witnesses that
	$V_{\alpha}$ is type 2 for all $\alpha \in (0, 1]$.

	Now let $\alpha > \frac12$ and consider the following two PDBs $\D_1$ and 
	$\D_2$, with probability measures $P_1$ and $P_2$, respectively, over 
	distinct facts $f$ and $g$:

	\begin{figure}[H]
		\begin{tabular}{lcc}\toprule
			$D$			& $\set{ f }$	& $\set{ g }$	\\\midrule
			$P_1( D )$	& $\alpha$ 		& $1-\alpha$ \\\bottomrule
		\end{tabular}
		\qquad	
		\begin{tabular}{lcc}\toprule
			$D$			& $\set{ f }$	& $\set{ g }$	\\\midrule
			$P_2( D )$	& $1-\alpha$ 	& $\alpha$ \\\bottomrule
		\end{tabular}
		\caption{Definition of $\D_1$ and $\D_2$ (given as an explicit list of
		the instances together with their probabilities).}
	\end{figure}

 	Then $D_{ \D_1 ,\alpha } = \set{ f }$ and $D_{ \D_2, \alpha } = \set{g}$. 
	Suppose $q$ exists such that $P_V = P \after q$ for every input PDB $\D = (
	\DB, \DDB, P )$. Let $\DDDD_{ f } \coloneqq \set{ D \in \DB_V \with f \in D 
	}$. Then $P_{ \D_1 } \big( q^{-1}( \DDDD_{f} ) \big) = 1$ implies $\set{ f } 
	\in q^{-1}( \DDDD_{ f } )$. On the contrary, 
	$P_{ \D_2 }\big( q^{-1}( \DDDD_{f} ) \big) = 0$, implies $\set{ f } \notin 
	q^{-1} ( \DDDD_{f} )$, a contradiction. Thus, $V_{\alpha}$ is not of type 3.
\end{proof}

Together, we have that $\VT1 \supsetneq \VT2 \supsetneq \VT3 \supsetneq
\VT4$. Before closing this section, let us highlight two key insights of the
arguments used in the examples for \cref{pro:1not2,pro:2not3}. In essence, we
separated type $2$ from type $3$ by arguing about the structure of possible 
worlds without really taking probabilities into account. For the separation of
type $1$ from type $2$, we argued about the structure of the probability 
measure instead. This highlights again the different levels within the
hierarchical structure of a PDB (fact level---instance level---probability
space level) that the views of the different classes operate on.

\section{Conclusions}\label{sec:conclusion}
We introduce the notion of standard PDBs, for which we rigorously describe how 
to construct suitable measurable spaces for infinite probabilistic databases,
completing the picture of~\cite{GroheLindner2019}. The viability of this model
as a general and unifying framework for finite \emph{and infinite} databases is
supported by the well-definedness and compositionality of (typical) query 
semantics. Other kinds of PDB queries embed into the framework as well.

It is currently open, whether, and if so, how more in-depth results on point
processes can be used for probabilistic databases, for example, to perform
open-world query answering. Also, while we focused on relational algebra and
aggregation, the queries of \cref{sec:beyond} deserve a systematic treatment 
in their own right in infinite PDBs.

\section*{Acknowledgments}\label{sec:acknowledgments}
We thank Sam Staton, who initially brought the correspondence of probabilistic
databases and finite point processes to our attention in personal 
correspondence.

This work is funded by \href{https://www.dfg.de}{\emph{Deutsche
Forschungsgemeinschaft} (DFG, German Research Foundation)} under grants
\href{https://gepris.dfg.de/gepris/projekt/412400621}{GR 1492/16-1} and
\href{https://gepris.dfg.de/gepris/projekt/282652900}{GRK 2236 (UnRAVeL)}.

\bibliographystyle{alphaurl}
\newcommand{\etalchar}[1]{$^{#1}$}

\appendix
\section{Topologies and Metric Spaces}\label{app:topo}

The brief account of topological notions is based upon~\cite{Willard2004}.

Let $\XX$ be some set. A \emph{topology} on $\XX$ is a family $\OOO$ of subsets
of $\XX$ that contains both $\emptyset$ and $\XX$, and is closed under
finite intersections and (arbitrary, i.\,e.\ possibly uncountable) unions. The
pair $(\XX,\OOO)$ is called a \emph{topological space}. The sets in $\OOO$ are
called \emph{open} sets, and their complements (with respect to $\XX$) are
called \emph{closed} sets.

For $\XXXX \subseteq \XX$, the \emph{closure} of $\XXXX$ is the intersection of
all closed sets that contain $\XXXX$. A set $\XXXX$ is called \emph{dense} in
$(\XX,\OOO)$ if its closure is $\XX$. The topological space $(\XX,\OOO)$ is
called \emph{separable} if there exists a countable subset of $\XX$ that is 
dense in $(\XX,\OOO)$.

Let $\XX$ be a non-empty set and let $d \from \XX \times \XX \to \RR$ be
a function satisfying
\begin{itemize}
	\item $d(X,Y) \geq 0$ with $d(X,Y) = 0$ if and only if $X = Y$;
	\item $d(X,Y) = d(Y,X)$; and
	\item $d(X,Z) \leq d(X,Y) + d(Y,Z)$ (the \emph{triangle inequality})
\end{itemize}
for all $X,Y,Z \in \XX$. Then $d$ is called a \emph{metric} on $\XX$ and the
pair $(\XX,d)$ is called a \emph{metric space}. For $\epsilon > 0$, the set
\[
	B_{\epsilon}( X ) \coloneqq \set[\big]{ Y \in \XX \with d(X,Y) < \epsilon }
\]
is the \emph{open ball of radius $\epsilon$ around $X$}. The \emph{metric
topology} on $\XX$ with respect to $d$ is the topology $\OOO$ on $\XX$ that is
defined by
\[
	\OOO \coloneqq \set[\big]{ \XXXX \subseteq \XX \with \text{ for all } X \in
	\XXXX \text{ there exists some } \epsilon > 0 \text{ such that }
	B_{\epsilon}(X) \subseteq \XXXX }\text.
\]
A sequence $(X_1, X_2, \dots)$ of elements in a metric space $(\XX,d)$ is
called a \emph{Cauchy sequence} if $\lim_{i \to \infty} d(X_i,X_{i+1}) = 0$.
A sequence $(X_1, X_2, \dots)$ in $(\XX,d)$ \emph{converges} to an element
$X \in \XX$ in $(\XX,d)$ if for all $\epsilon \in (0,\infty)$ there exists 
some $i$, large enough, such that $d( X,X_i ) < \epsilon$. A metric $d$ on
$\XX$ is called \emph{complete}, if all Cauchy sequences over $(\XX,d)$ 
converge to elements of $\XX$. A metric space is called \emph{complete}, if
its metric is complete. In a metric space, a set $\XXXX$ is dense (with respect
to the metric topology) if and only if for all $X \in \XX$ and all $\epsilon >
0$ there exists some $Y \in \XX$ such that $d(X,Y) < \epsilon$.

A topological space $(\XX,\OOO)$ is called \emph{metrizable} if there exists a
metric $d$ on $\XX$ such that $\OOO$ is the corresponding metric topology. If
such $d$ exists that is additionally complete, the space is called
\emph{completely metrizable}. A \emph{Polish space} is a completely metrizable,
separable topological space.

\end{document}